\documentclass[a4paper,10pt]{article}
\usepackage[a4paper, total={6in, 10in}]{geometry}
\usepackage[utf8x]{inputenc}
\usepackage{amssymb,amsmath,amsthm}
\usepackage{hyperref}

\usepackage{graphicx}
\usepackage{color}
\usepackage{comment}
\graphicspath{{figs/}}
\newtheorem{theorem}{Proposition}

\newtheorem{cor}{Corollary}
\newtheorem{remark}{Remark}
\newcommand{\xx}{\mathbf{x}}
\newcommand{\Z}{\mathbf{Z}}
\newcommand{\E}{\mathbf{E}}
\newcommand{\R}{\mathbb{R}}
\newcommand{\D}{D}
\newcommand{\bloc}{B}
\newcommand{\sfold}{\mathcal{S}}
\newcommand{\ii}{\mathbf{i}}
\newcommand{\jj}{\mathbf{j}}

\newcommand{\myv}{\mathbf{v}}
\newcommand{\mymu}{\boldsymbol{\mu}}
\newcommand{\myV}{\mathbf{V}}

\newcommand{\mybeta}{\boldsymbol{\beta}}
\newcommand{\covnoise}{\Sigma_{\varepsilon}}
\newcommand{\covbeta}{\Sigma_{\beta}}

\newcommand{\covsum}{\Sigma}

\newcommand{\gp}{\xi}
\newcommand{\cgp}{\eta}
\newcommand{\noise}{\boldsymbol{\varepsilon}}
\newcommand{\kweights}{\boldsymbol{\lambda}}
\newcommand{\fibeta}{\mathcal{I}}
\newcommand{\fibetat}{\mathcal{I}_{-\new}}
\newcommand{\gencm}{A}
\newcommand{\genm}{\mathbf{m}}
\newcommand{\genQ}{P}

\title{
Fast calculation of Gaussian Process multiple-fold cross-validation residuals and their covariances
}
\author{David Ginsbourger\footnote{Institute of Mathematical Statistics and Actuarial Science, Department of Mathematics and Statistics, University of Berne, Alpeneggstrasse 22, CH-3012 Bern, Switzerland. \newline
	E-mail: \href{david.ginsbourger@unibe.ch}{david.ginsbourger@unibe.ch}} 
\ and C\'edric Sch\"arer\\
Department of Mathematics and Statistics, \\
University of Berne, Switzerland
}

\begin{document}
\maketitle

\thispagestyle{empty}
\begin{abstract}
We generalize fast Gaussian process leave-one-out formulae to multiple-fold cross-validation, highlighting in turn the covariance structure of cross-validation residuals in both Simple and Universal Kriging frameworks. We illustrate how resulting covariances affect model diagnostics. We further establish in the case of noiseless observations that correcting for covariances between residuals in cross-validation-based estimation of the scale parameter leads back to MLE. Also, we highlight in broader settings how differences between pseudo-likelihood and likelihood methods boil down to accounting or not for residual covariances. The proposed fast calculation of cross-validation residuals is implemented and benchmarked against a naive implementation. Numerical experiments highlight the accuracy and substantial speed-ups that our approach enables. However, as supported by a discussion on main drivers of computational costs and by a numerical benchmark, speed-ups steeply decline as the number of folds (say, all sharing the same size) decreases. An application to a contaminant localization test case illustrates that grouping clustered observations in folds may help improving model assessment and parameter fitting compared to Leave-One-Out. Overall, our results enable fast multiple-fold cross-validation, have direct consequences in model diagnostics, and pave the way to future work on hyperparameter fitting and on the promising field of goal-oriented fold design.       
\end{abstract}

\section{Introduction and notation}

Gaussian process (GP) models are at the heart of a number of prominent methods in spatial statistics, machine learning, and beyond. Properly validating such models and fitting underlying hyperparameters thus appear as crucial and impactful endeavours. 
It is often the case that, for reasons of data scarcity or other, the latter need to be conducted based on a unique data set. Cross-validation is commonly used in such a context, not only in GP modelling but also as a general approach to assess statistical models and tune parameters in a vast class of prediction algorithms. Fast leave-one-out cross-validation formulae are known for some models and approaches including GP prediction. However leave-one-out is known to suffer some pitfalls, and multiple-fold cross-validation is increasingly preferred over it in broader contexts, notably in the realm of machine learning.

\medskip 

We focus here on cross-validation in the context of GP modelling, and more specifically on adapting fast formulae for leave-one-out to multiple-fold cross-validation and exploring how these can be efficiently exploited in topics such as model diagnostics and covariance hyperparameter fitting. One of the essential bricks of the proposed approaches turns out to be closed-form formulae for the covariance structure of cross-validation residuals. Our main results are presented both in the Simple Kriging and Universal Kriging frameworks. They contribute to a lively stream of research on cross-validation that pertains to GP modelling, and also to further related classes of statistical models. 

\medskip 

Leave-one-out and multiple-fold cross-validation in the GP framework were tackled as early as in \cite{Dubrule1983}, and some further references as well as a more recent account of their use in machine learning is given in Chapter~5 of \cite{Rasmussen.Williams2006}. 
On a different note, cross-validation has lead to a number of investigations in the context of spline modelling, for instance in \cite{Golub1979} where generalized cross-validation was used for the selection of ridge (regularization) parameter.   
In the realm of GP modelling for computer experiments, cross-validation was tackled as early as in \cite{Currin1989} (with a pioneering discussion on links to MLE), and leave-one-out was investigated in \cite{Bachoc2013} and found to deliver a valuable alternative to MLE in misspecified cases.  
A general survey of cross-validation procedures for model selection is presented in \cite{Arlot2010}. 
Cross-validation for model selections is further tackled in \cite{Zhang2015}. 
\cite{LeGratiet.etal2015} proposed cokriging-based sequential design strategies using fast cross-validation for multi-fidelity computer codes.
\cite{Fischer2016} studied model selection for GP Regression by Approximation Set Coding. \cite{Martino2017} suggested new probabilistic cross-validation Estimators for Gaussian Process Regression. Cross-validation strategies for data with temporal, spatial, hierarchical, or phylogenetic structure were investigated in \cite{Roberts2017}. 
New prediction error estimators were defined in \cite{Rabinowicz2020}, where novel model selection criteria in the same spirit as AIC and Mallow's $C_{p}$ were suggested.
Leave-One-Out cross-validation for Bayesian model comparison in large data was recently tackled in \cite{Magnusson2020}. \cite{Kaminsky2021} introduced an efficient batch multiple-fold cross-validation Voronoi adaptive sampling technique for global surrogate modeling.  
\cite{Rabinowicz2020a} introduced a bias-corrected cross-validation estimator for correlated data. 
\cite{Bates2023} tackles the question of what one actually does estimate in cross-validation. In contrast with the random design settings considered in the latter as in many theoretical works pertaining to cross-validation, our baseline settings here will assume a fixed design.  

\bigskip 

Leave-One-Out (LOO) cross-validation consists in predicting at each of the observation locations when (virtually) removing the corresponding observation from the data set, and then comparing obtained predictions to left-out observations. 
As a first example, we consider a case of GP prediction for a one-dimensional test function recalled in Appendix~\ref{app:example} observed at a regular design, here a 10-point subdivision of $[0,1]$, and a GP model assuming a Mat\'ern 5/2 stationary covariance kernel. The test function is represented in black in Figure~\ref{xiong_preds}, along with the GP predictor (blue line) and Leave-One-Out (LOO) predictions at the design locations (red points). Modelling is performed with the \textrm{DiceKriging} package \cite{Roustant.etal2012} and cross-validation relies on a fast implementation using the \textrm{cv} function newly available in the \textrm{DiceKriging} package (version 1.6.0). This simple example in a misspecified case illustrates how LOO residuals may be more representative of actual prediction errors than the built-in GP prediction standard deviation, the latter being in this case not depending on the actual observations and calculated under a questionable stationarity assumption. 

\begin{figure}[h!]
	\begin{center}
		\includegraphics[width=.65\linewidth]{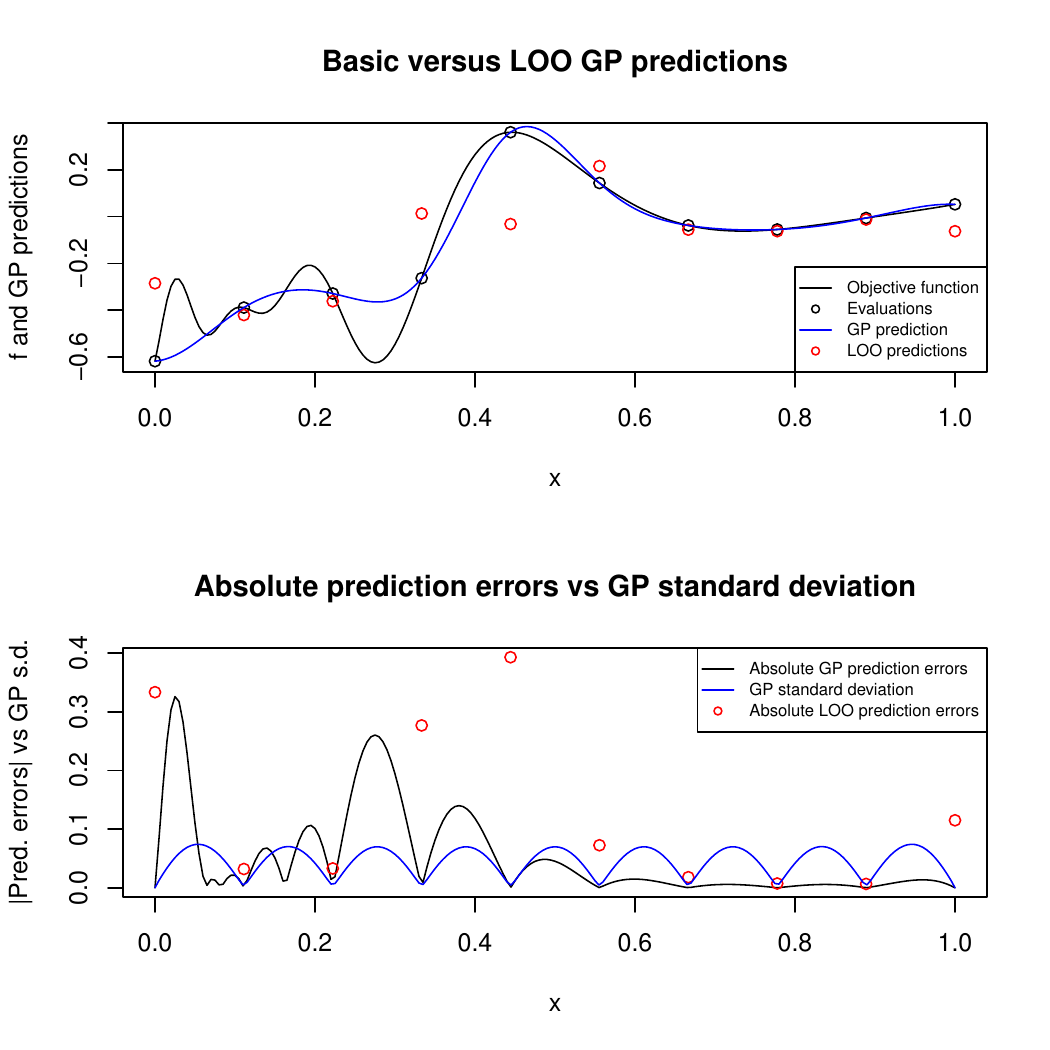}
		\caption{On the upper panel, GP mean predictor (blue line)  of the test function (black line) defined by Equation~\ref{xiong} based on $10$ evaluations at a regular grid, LOO  cross-validation predictions (red points). Lower panel: absolute prediction errors associated with GP (black line) and LOO (red point) predictions, and GP prediction standard deviation (in blue).}
		\label{xiong_preds}
	\end{center}
\end{figure}
It has in fact been found as early as in 1983 (See Dubrule's seminal paper \cite{Dubrule1983}) that the LOO prediction errors and variance could be calculated quite elegantly and efficiently by relying on the inverse of a certain matrix (the main focus in in \cite{Dubrule1983} being on a variogram matrix extended by monomial basis functions) and quantities derived thereof. We work  here in terms of covariances and related matrices as in \cite{Bachoc2013}, where fast LOO formulas were obtained for Simple Kriging in the square-integrable case. 

An elegant by-product formula for the square norm of LOO residulas has been leveraged in \cite{Bachoc2013} for hyperparameter estimation in the stationary case, namely by minimizing this quantity as a function of covariance parameters (excluding the ``variance'' parameter).  
In particular, numerical experiments conducted in \cite{Bachoc2013} suggested that leave-one-out error minimization is more robust than maximum likelihood in the case of model misspecification, a common situation in practice.
Yet, depending on the experimental design, 
the errors obtained by leaving one point at a time may be far from representative of generalization errors, as illustrated in Appendix~\ref{ex2_fastmfcv}. While remaining with a prescribed design, multiple-fold cross-validation allows to mitigate shortcomings of leave-one-out by considering situations where larger parts of the design (i.e., several points) are removed at a time.  

\bigskip 

Our primary aim in the present paper is to generalize fast-leave-one-out formulae to multiple-fold cross-validation, hence facilitating their use in both contexts of model diagnostics and hyperparameter optimization. We derive such formulae and highlight in turn the covariance structure of cross-validation residuals, a crucial ingredient in the proposed diagnostics and findings pertaining to parameter fitting.
Furthermore, following ideas already sketched in the seminal paper of Dubrule \cite{Dubrule1983}, we present extensions of our results to the case of Universal Kriging. 
The obtained results clarify some links between cross-validation and maximum likelihood, and also open new perspectives regarding the design of folds and ultimately also of the underlying experimental design points.  

\bigskip

Section~\ref{secbasics} presents the consisdered class of statistical models and recalls the basics of Simple and Universal Kriging under Gaussian Process assumptions, with some brief detours in \ref{basicsUK} on how Gaussian linear and ridge regression may fit into this framework as well as on Bayesian interpretations of Universal Kriging and ridge regression. Prediction equations are revisited in Section~\ref{sec:transductive} in transductive settings thanks to Schur complement approaches, hence setting the decor for fast cross-validation.      
Section~\ref{secfastMFCV} is dedicated to fast multiple-fold cross-validation, starting in Section~\ref{sec:genericGV} with generic results that apply in Gaussian vector conditioning. We then present in Section~\ref{SKcase} our main results on the fast calculation and the joint probability distribution of cross-validation residuals in the Simple Kriging case. 
The latter are then extended to Universal Kriging settings in Section~\ref{Extensions}, notably retrieving as a particular case fast cross-validation results for linear regression presented in \cite{Zhang1993}. In Section~\ref{csq_fitting} we present consequences of the main results in the context of model fitting. Section~\ref{diag} focuses on graphical diagnostics via pivotal statistics for cross-validation residuals, while the two subsequent sections deal with  cross-validation-based covariance parameter estimation. More specifically, Section~\ref{scale} is devoted to the estimation of the scale parameter; Section~\ref{further} then pertains to the estimation of further hyperparameters building upon either the norm or the (pseudo-)likelihood of multiple-fold cross-validation residuals. Section~\ref{comp} mostly consists in numerical experiments illustrating the accuracy and speed-ups offered by the closed-form formulae (Section~\ref{speedup}), and some considerations regarding associated computational complexities (Section~\ref{complexity}). 
Section~\ref{conta} presents an application to a contaminant localization test case and illustrate that grouping clustered observations in folds lead to improved model assessment and parameter fitting compared to Leave-One-Out.
The ultimate Section~\ref{discussion} is a discussion opening perspectives on how the established results could be further exploited to improve diagnostics and parameter estimation procedures in GP modelling. 

\section{Models and transductive prediction settings}
\label{secbasics}
\subsection{Of Universal Kriging and ridge regression}
\label{basicsUK}

Throughout the article we consider statistical models of the form 
\begin{equation}
Z_{i}=\gp(\xx_i)+\varepsilon_i \ \ \ (\xx_i \in D, i \in \{1,\dots,n\}),
\end{equation}
with $\gp(\xx)=\mu(\xx)+\cgp(\xx)$ ($\xx\in D$) where $\cgp$ is a centred GP with covariance kernel $k$ and $\mu$ is a trend function, and $\noise\sim \mathcal{N}(\mathbf{0}, \covnoise)$ independently of $\cgp$. In some cases, $\mu$ is known up to some coefficients: unknown constant (\textit{Ordinary Kriging} settings) or linear combination of basis functions with unknown coefficients (\textit{Universal Kriging} settings). When $\mu$ is known, one speaks of Simple Kriging settings. While several variants of Kriging can be defined beyond the second order case (without assuming square-integrability of $\cgp$), we restrict the exposition here to the broadly used Gaussian case. However we do not require $k$ to be stationary. Let us remark that throughout the following exposition on Kriging/GP prediction, the kernel (including hyperparameters) is considered as given. In practice, the equations are also employed with plugged in estimates of hyperparameters. Full Bayesian approaches where the uncertainty on kernel hyperparameters is further propagated to the predictions are beyond the considered scope. However in Sections~\ref{csq_fitting} and~\ref{conta} we tackle the issue of model validation and hyperparameter fitting via the established cross-validation results.  

\medskip 

We consider a linear trend case that encompasses all of the above, and also has some interesting connections to linear (ridge) regression as we will see below. We hence assume $\mu$ to be defined by
\begin{equation}
\mu(\xx)=\sum_{j=1}^p \beta_j f_j(\xx) \ \ \ (\xx \in D), 
\end{equation} 
where the $f_{j}$ ($1\leq j \leq p$, $1\leq p$) 
are prescribed basis functions and $\beta_j$ ($1\leq j \leq p$) are real-valued coefficients. 
Simple Kriging can be retrieved by taking $p=1$ with $\beta_1=1$ and $f_1$ arbitrary, and Universal Kriging with $p\geq 1$ and $\mybeta=(\beta_1,\dots,\beta_p)$ unknown (Ordinary Kriging is a special case with $p=1$ and $f_1\equiv 1$). 
In Simple Kriging, a predictor of $\gp$ of the form 
$\widehat{\gp}(\xx)=\mu(\xx)+\kweights(\xx)^{\top}\Z \ (\xx \in D)$
is sought, where $\Z=(Z_{\xx_1},\dots,Z_{\xx_n})^{\top}$. 
$\kweights(\xx)$ is determined by minimizing the residual variance 
\begin{equation}
\label{quad_risk}
\operatorname{Var}[\gp(\xx)-\widehat{\gp}(\xx)]=k(\xx,\xx)+\kweights(\xx)^{\top}(K+\covnoise)\kweights(\xx)-2 \mathbf{k}(\xx)^{\top}K\kweights(\xx), 
\end{equation}
where  $K=(k(\xx_i,\xx_j))_{i,j \in \{1,\dots, n\}}$ and $\mathbf{k}(\xx)=(k(\xx,\xx_i))_{i \in \{1,\dots,n\}}$. 
In Universal Kriging, the predictor is of the form
\begin{equation*}
\widehat{\gp}(\xx)=\kweights(\xx)^{\top}\Z \ \ \ (\xx \in D),
\end{equation*}
and $p$ linear constraints $f_{j}(\xx)=\sum_{i=1}^n \lambda_i(\xx) f_{j}(\xx_j)$ for $j \in \{1,\dots,p\}$ are added to the minimization of Eq.~\ref{quad_risk} to ensure unbiasedness of $\widehat{\gp}(\xx)$. 
Minimizing Eq.~\ref{quad_risk} under these unbiasedness constraints can be addressed by Lagrange duality, leading to a linear problem featuring a $p$-dimensional Lagrange multiplier $\boldsymbol{\ell}(\xx)$: 
\begin{equation}
\label{UK_system}
\left( \begin{matrix} 
K+\covnoise & F \\
F^{\top} & 0 
\end{matrix} \right)
\left( \begin{matrix} 
\boldsymbol{\lambda}(\xx) \\
\boldsymbol{\ell}(\xx)
\end{matrix} \right)
=
\left( \begin{matrix} 
\mathbf{k}(\xx) \\
\mathbf{f}(\xx)
\end{matrix} \right),
\end{equation}
where $F=(f_{j}(\xx_i))_{1\leq i \leq n, 1\leq j \leq p} \in \R^{n\times p}$ 
and $\mathbf{f}(\xx)=(f_1(\xx),\dots,f_p(\xx))^{\top} \in \R^p$. 
Assuming that $\covsum=K+\covnoise$ and $F^{\top}\covsum^{-1}F$ are invertible, solving for $\boldsymbol{\lambda}(\xx)$ delivers the Universal Kriging predictor, that can ultimately be written, as further detailed in Appendix~\ref{basicsUKold}, as
\begin{equation}
\label{UKpred}
\widehat{\gp}(\xx)= \mathbf{f}(\xx)^{\top}\widehat{\mybeta} + \mathbf{k}(\xx)^{\top}\covsum^{-1}(\Z-F\widehat{\mybeta}) \ \ \ (\xx \in D)
\end{equation}
where $\widehat{\mybeta}=\fibeta^{-1} F^{\top} \covsum^{-1}\Z$ with $\fibeta=F^{\top}\covsum^{-1}F$. 
In Simple Kriging with $\mu(\xx)=\mathbf{f}(\xx)^{\top}\mybeta$, the predictor has the same form as Eq.~\ref{UKpred} yet with $\mybeta$ instead of $\widehat{\mybeta}$.   
Furthermore, the Universal Kriging residual covariance writes, for arbitrary $\xx, \xx' \in D$:
\begin{equation}
\label{UK_cov}
\begin{split}
&\operatorname{Cov}[\gp(\xx)-\widehat{\gp}(\xx),\gp(\xx')-\widehat{\gp}(\xx')]=k(\xx,\xx')-
\mathbf{k}(\xx)^{\top}\covsum^{-1}\mathbf{k}(\xx') \\
&+(\mathbf{f}(\xx)-F^{\top}\covsum^{-1}\mathbf{k}(\xx))^{\top}\fibeta^{-1}(\mathbf{f}(\xx')-F^{\top}\covsum^{-1}\mathbf{k}(\xx')). 
\end{split}
\end{equation}
The Simple Kriging residual covariance boils down to the right handside of the first line above.

On a different note, it is worth noting that when setting $\cgp\equiv 0$ (and hence $\gp(\xx)=\sum_{j=1}^p \beta_j f_j(\xx)$) one retrieves thereby the equations of (Generalized) Least Squares. 

\medskip

We will now review the Bayesian approach to Universal Kriging \cite{Omre.Halvorsen1989,Handcock.Stein1993,Helbert}, that will in turn give us an occasion to incorporate ridge regression into the discussion and will prove practical in further developments throughout the paper.  
Let us assume to this end that $\mybeta$ is endowed with a Gaussian prior distribution $\mathcal{N}(\mathbf{0}, \covbeta)$ where $\covbeta$ is an invertible covariance matrix. 
Let us first consider $\covbeta$ as fixed. The posterior distribution of $\mybeta$ knowing $\Z$ is 
Gaussian with $\mathbb{E}[\mybeta | \Z]=\widehat{\mybeta}_{\text{MAP}}:=(\covbeta^{-1}+F^{\top}\covsum^{-1}F)F^{\top} \covsum^{-1}\Z$ 
and 
\begin{equation}
\label{BUQ_covbeta}
\operatorname{Cov}[\mybeta | \Z]=  (\covbeta^{-1} + F^{\top}\covsum^{-1}F)^{-1}
\end{equation}
whereof, for $\xx, \xx' \in D$, $\mathbb{E}[\gp(\xx) | \Z]=\mathbf{f}(\xx)^{\top}\widehat{\mybeta}_{\text{MAP}} + \mathbf{k}(\xx)^{\top}\covsum^{-1}(\Z-F\widehat{\mybeta}_{\text{MAP}})$ and
\begin{align}
\label{BUQ_covres}
&\operatorname{Cov}[\gp(\xx),\gp(\xx')|\Z]=k(\xx,\xx)-
\mathbf{k}(\xx)^{\top}\covsum^{-1}\mathbf{k}(\xx') \\
&+(\mathbf{f}(\xx)-F^{\top}\covsum^{-1}\mathbf{k}(\xx))^{\top}\operatorname{Cov}[\mybeta | \Z](\mathbf{f}(\xx')-F^{\top}\covsum^{-1}\mathbf{k}(\xx')). \notag
\end{align}
With $\covbeta^{-1}\to \mathbf{0}$, we hence have that
 $\operatorname{Cov}[\mybeta | \Z] \to  \fibeta^{-1}$ and so 
 $\mathbb{E}[\gp(\xx) | \Z] \to \widehat{\gp}(\xx)$
 and $\operatorname{Cov}[\gp(\xx),\gp(\xx')|\Z] \to \operatorname{Cov}[\gp(\xx)-\widehat{\gp}(\xx),\gp(\xx')-\widehat{\gp}(\xx')]$. 

\medskip 

We now turn to linear (ridge) regression and present how the resulting equations can be obtained in similar settings. 
In ridge regression, say as classically with $\covnoise=\sigma^2 I_{n},$ one estimates $\mybeta$ in $\Z=F\mybeta+\boldsymbol{\varepsilon}$ via $\widehat{\mybeta}_{\lambda}=(F^{\top}F+\lambda I_{n})^{-1}F^{\top}\Z$, where $\lambda >0$ is a regularization parameter. The canonical way to introduce this estimator is by replacing the Ordinary Least Squares minimization by the minimization of its penalized counterpart 
$$
\mybeta \mapsto ||\Z-F\mybeta||^2_{2} + \lambda ||\mybeta||^2_{2}.
$$
Yet it is well-known that $\widehat{\mybeta}_{\lambda}$ can also be seen as Maximum A Posterior estimator in a Bayesian framework, under the prior distribution $\mybeta\sim \mathcal{N}(\mathbf{0}, \gamma^2 I_{p})$, where $\gamma^2=\frac{\sigma^2}{\lambda}$. In such a framework, one has indeed $\Z\sim \mathcal{N}(\mathbf{0}, \gamma^2 FF^{\top}+\sigma^2 I_n)$ and 
$\mybeta | \Z=\mathbf{z}
\sim \mathcal{N}(F^{\top}(FF^{\top}+\lambda I_p)^{-1}\mathbf{z}, \gamma^2 I_{p}-\gamma^2F^{\top}(FF^{\top}+\lambda I_n)^{-1}F)$. From there one can obtain for instance by using Sherman-Morisson's formula that 
$$F^{\top}(FF^{\top}+\lambda I_p)^{-1}=(F^{\top}F+\lambda I_p)^{-1}F^{\top},$$ 
whereof $\mathbb{E}[\mybeta| \Z]=(F^{\top}F+\lambda I_p)^{-1}F^{\top}\Z=\widehat{\mybeta}_{\lambda}$,
which coincides in the considered Gaussian case with the MAP estimator. Interestingly, for $F$ full column-ranked and $\gamma \to \infty$, the posterior distribution is proper and one recovers as posterior expectation and MAP the Ordinary Least squares Estimator $\widehat{\mybeta}=(F^{\top}F)^{-1}F^{\top}\Z$.

\subsection{Transductive GP prediction via Schur complements}
\label{sec:transductive}

\newcommand{\xnew}{\xx_{\ii_{o}}}
\newcommand{\xold}{\xx_{\jj_{o}}}
\newcommand{\new}{\ii_{o}}
\newcommand{\old}{\jj_{o}}
\newcommand{\covsums}{\covsum_{\star}}
\newcommand{\augcms}{M_{\star}}
\newcommand{\augcm}{M}
\newcommand{\Zaug}{\Z^{+}}
\newcommand{\Ktilde}{\widetilde{K}}
\newcommand{\npred}{m}
\newcommand{\cvpred}[1]{\widehat{\Z}^{(-#1)}[#1]}

One speaks of transductive settings in cases where the point(s) at which one wishes to predict are fixed in advance. We consider in this section the case where $\gp(\xx_i)$ (resp. $Z_i$) are to be predicted for $\npred \leq n-1$ distinct values of $i$ ($\npred \leq n-p$ in the case of Universal Kriging) from the remaining $n-\npred$ observations $Z_j$. Without loss of generality, we will assume that the indices at which the predictions are to be made are $\ii_{o}=(1,\dots,\npred)$ and denote by $\jj_{o}=(\npred+1,\dots,n)$ where $\npred\leq n-1$ the remaining ones. For any vector $\myv \in \R^n$, matrix $A\in \R^{n\times n}$, and arbitrary vectors of ordered indices $\ii, \jj$,  we denote here and in the following by $\myv[\ii]$ the subvector of $\myv$ with indices in $\ii$ and by $A[\ii,\jj]$ the block extracted from $A$ with corresponding indices, using the in turn the convention $A[\ii]=A[\ii,\ii]$. Similarly, we use the convention $k(\xx)=k(\xx,\xx)$ for $\xx \in D$. 

\medskip

In Simple Kriging settings (here $\mu\equiv 0$), it is common knowledge that $\widehat{\gp}(\xnew)$ and the associated residual covariance matrix $\operatorname{Cov}(\gp(\xnew)-\widehat{\gp}(\xnew))$ can be obtained elegantly based on manipulating 
\begin{equation}
\covsums=\left( \begin{matrix} k(\xnew)  &  k(\xnew, \xold) \\  k(\xold, \xnew) &   k(\xold)+\covnoise[\old] \end{matrix} \right).
\end{equation} 
Assuming indeed that $\covsums$ is invertible, we obtain by bloc inversion formula (See Appendix~\ref{app_schur}, in particular Theorem~\ref{Schur_inv}) 
that $\covsums^{-1}[\new]=(k(\xnew)-k(\xnew, \xold)\covsum[\old]^{-1}k(\xold, \xnew))^{-1}$, whereof 
\begin{equation}
\label{invcondcov}
\operatorname{Cov}(\gp(\xnew)-\widehat{\gp}(\xnew))=(\covsums^{-1}[\new])^{-1}
\end{equation}
where $\covsums^{-1}[\new]$ stands for the $\npred\times\npred$ upper left bloc of $\covsums^{-1}$. 
Using now the upper right bloc $\covsums^{-1}[\new,\old]$ of $\covsums^{-1}$ and Eq.~\ref{invcondcov}, we obtain similarly that $\covsums^{-1}[\new,\old]=-\covsums^{-1}[\new]k(\xnew, \xold)\covsum[\xold]^{-1}$, so 
\begin{equation}
\widehat{\gp}(\xnew)=-(\covsums^{-1}[\new])^{-1}\covsums^{-1}[\new,\old]\Z[\jj_o].
\end{equation}  

Applying now under invertibility assumption a similar bloc inversion to 
\begin{equation}
\covsum=
\left( \begin{matrix} \covsum[\new]  &  \covsum[\new,\old] \\  \covsum[\old,\new] &   \covsum[\xold] \end{matrix} \right)
=\left( \begin{matrix} k(\xnew)+\covnoise[\new]  &  k(\xnew, \xold)+\covnoise[\new,\old] \\  k(\xold, \xnew)+\covnoise[\old,\new] &   k(\xold)+\covnoise[\xold] \end{matrix} \right),
\end{equation} 
we get 
$\covsum^{-1}[\new]=( \covsum[\new]- \covsum[\new,\old]\covsum[\xold]^{-1}\covsum[\old,\new])^{-1}$ 
and
$\covsum^{-1}[\new,\old]=-\covsum^{-1}[\new]\covsum[\new, \old]\covsum[\old]^{-1}$ 
whereof the BLUP of $\Z_{\ii_o}$ given 
$\Z_{\jj_o}$, denoted here and in the following by $\widehat{\Z}^{(-\ii_o)}[\ii_o]$, is given by 
\begin{equation}
\begin{split}
\cvpred{\new}=-(\covsum^{-1}[\new])^{-1}\covsum^{-1}[\new,\old]\Z[\old].
\end{split}
\end{equation}  

\newcommand{\Q}{Q}

Considering now the prediction residual $\E_{\new}:=\Z[\new]-\cvpred{\new}$, we get from the latter 
\begin{equation}
\label{MFCVres_SK}
\begin{split}
\E_{\new}&=(\covsum^{-1}[\new])^{-1}(\covsum^{-1}\Z)[\new],\\
&=(Q[\new])^{-1}(Q\Z)[\new],
\end{split}
\end{equation}  
with $\Q=\covsum^{-1}$.
Let us stress that $\cvpred{\new}$ departs from the most common Simple Kriging predictor via the treatment of noise, in the sense that in the considered transductive settings it can relevant to predict the noisy $\Z[\new]$ instead of $\gp(\xnew)$, and since a general noise covariance matrix is assumed for the sake of generality, this implies a straightforward adaptation of the blocs entering into play (i.e. blocs of $\covsum$ instead of those of $\covsums$). Of course, in case of a null observation noise (as often encountered in the realm of computer experiments), the two approaches above coincide. This is also the case if $\covsum$ is zero for line/column indices outside of $\new$. 

\medskip 
We will now see that under invertibility conditions, the Universal Kriging predictor and residual covariance too can be retrieved based on the inverse of a matrix, namely of 
\begin{equation}
\augcms=
\left( \begin{matrix} 
k(\xnew) & k(\xnew, \xold) & F[\new,] \\ 
k(\xold,\xnew) & k(\xold)+\covnoise[\xold] & F[\old,] \\
F[\new,]^{\top} & F[\old,]^{\top} & 0 
\end{matrix} \right),
\end{equation}
where $F[\new,]$ is the $\npred\times p$ matrix of basis functions evaluated at the prediction points. 
As we prove below, assuming invertibility of $\augcms$ and $\augcms[-\new]=\augcms[\npred+1:\npred+p]$, the Universal Kriging residual covariance matrix can then be plainly obtained as follows:
\begin{equation}
\label{covcondUK}
\operatorname{Cov}(\gp(\xnew)-\widehat{\gp}(\xnew))=(\augcms^{-1}[\new])^{-1}
\end{equation} 
Furthermore, the Universal Kriging predictor too writes similarly as in the Simple Kriging case:
\begin{equation}
\label{transductiveUK}
\widehat{\gp}(\xnew)=-(\augcms^{-1}[\new])^{-1}\augcms^{-1}[\new,\old]\Z[\old]
\end{equation} 
In fact, in the two last equations, what changed between Simple and Universal Kriging settings in simply the replacement of $\covsums$ by $\augcms$. 
Let us now present how this works. 
Assuming indeed that $\augcms$ and $\augcms[-\new]$ are invertible, we first apply Theorem~\ref{Schur_inv} from Appendix~\ref{app_schur} and obtain 
\begin{equation}
\label{SchurUK}
\augcms^{-1}[\new]^{-1}=k(\xnew)-[k(\xnew,\xold), F[\new,]]\augcms[-\new]^{-1}[k(\xnew,\xold), F[\new,]]^{\top}
\end{equation}

Now, provided $\covsum[\old]=k(\xold)+\covnoise[\old]$ and $\fibetat=F[\old,]^{\top}\covsum[\xold]^{-1}F[\old,]$ are invertible, we get using a variation of the bloc inversion formula (detailed in Remark~\ref{schur_variant}):
\begin{equation}
\label{invKtilde}
\augcms[-\new]^{-1}=
\left( \begin{matrix} 
\covsum[\old]^{-1} -\covsum[\old]^{-1}F[\old,]\fibetat^{-1}F[\old,]^{\top}\covsum[\old]^{-1}& \covsum[\old]^{-1}F[\old,]\fibetat^{-1} \\
\fibetat^{-1}F[\old,]^{\top}\Sigma[\old]^{-1} & - \fibetat^{-1} 
\end{matrix} \right)
\end{equation}

Expanding the product on the right handside of Eq.~\ref{SchurUK} using Eq.~\ref{invKtilde} then delivers 
\begin{align}
&(k(\xnew,\xold), F[\new,])\augcms^{-1}[\new]^{-1}(k(\xnew,\xold), F[\new,])^{\top} \notag \\ 
=& k(\xnew,\xold)\covsum[\old]^{-1}k(\xold,\xnew)\\
-&(F[\new,]-k(\xnew,\xold)\covsum[\old]^{-1}F)\fibetat^{-1}(F[\new,]-k(\xnew,\xold)\covsum[\old]^{-1}F)^{\top}, \notag
\end{align}
which establishes Eq.~\ref{covcondUK}. As for Eq.~\ref{transductiveUK}, it then follows from Theorem~\ref{Schur_inv}) and Eq.~\ref{invKtilde}
that 
\begin{equation*}
\begin{split}
\augcms^{-1}[\new,\old]\Z[\old]
&=-\augcms^{-1}[\new] (k(\xnew,\xold), F[\new,])\augcms[-\new]^{-1}\widetilde{\Z} \\
\text{whereof }-(\augcms^{-1}[\new])^{-1}\augcms^{-1}[\new,-\new]\widetilde{\Z}
&=k(\xnew,\xold)\covsum[\old]^{-1}(\Z-F\widehat{\mybeta}^{(-\new)})+F[\new,]\widehat{\mybeta}^{(-\new)},
\end{split}
\end{equation*}
with $\widehat{\mybeta}^{(-\new)}$ the GLS estimator of $\mybeta$ based on $\Z[\old]$. 

Applying now under suitable invertibility assumptions a similar bloc inversion to 
\begin{equation}
\label{defM}
\augcm=
\left( \begin{matrix} \covsum  & F \\ F^{\top}  &  \textbf{0} \end{matrix} \right),
\end{equation} 
we get the BLUP $\widehat{\Z}^{(-\ii_o)}[\ii_o]$ of $\Z_{\ii_o}$ given $\Z_{\jj_o}$ under Universal Kriging settings as
\begin{equation}
\cvpred{\new}=-(\augcm^{-1}[\new])^{-1}\augcm^{-1}[\new,\old]\Z[\old].
\end{equation}  

\newcommand{\Qt}{\widetilde{\Q}}

Considering further the associated prediction residual $\E_{\new}:=\Z[\new]-\cvpred{\new}$, we finally get  
\begin{equation}
\begin{split}
\label{MFCVres_UK}
\E_{\new}&=(\augcm^{-1}[\new])^{-1}(\augcm^{-1}[1:n]\Z)[\new],\\
&=(\Qt[\new])^{-1}(\Qt\Z)[\new],
\end{split}
\end{equation} 
with $\Qt=\Q-\Q F (F^{\top}\Q F)^{-1} F^{\top} \Q$.

\section{Fast multiple-fold cross-validation}
\label{secfastMFCV}

In this section we leverage transductive GP equations in settings when subsets of the $n$ observations are left out, develivering in turn the joint distribution of multiple-fold cross validation residuals together with a fast computation approach. The key is that the previous results of Equation~\ref{MFCVres_SK} and Equation~\ref{MFCVres_UK} carry over to arbitrary ways of separating the data set into a left out and remaining observations (i.e. between test and learning sets).  
Throughout the following we consider vectors of strictly ordered indices from $\{1,\dots, n\}$, and denote by $\sfold$ the set of all such index vectors. Since these index vectors are completely characrterized by non-empty subsets of $\{1,\dots,n\}$, there are $2^n-1$ elements in $\sfold$. Elements $\ii \in \sfold$, that may also be thought of as subsets of $\{1,\dots, n\}$, are called \textit{folds}.

\subsection{The generic Gaussian vector case}
\label{sec:genericGV}

We now consider generic Gaussian cross-validation settings and deterministic functions arising in this context. 
$\myv \in \R^n$ and $\gencm \in \R^{n\times n}$ stand for a vector and a symmetric positive definite matrix, respectively. For any $\ii \in \sfold$, we defined the following cross-validation residual mapping:
\begin{align}
e_{\ii}(\cdot; \genm,\gencm): \myv \in \R^n \to e_{\ii}(\myv; \genm,\gencm)=\myv[\ii]-\genm[\ii] - 
\gencm[\ii,-\ii]\gencm[\ii]^{-1}(\myv[-\ii]-\genm[-\ii]),
\end{align}
delivering for $\myv \in \R^n$ the deterministic counterpart of the cross-validation residual 
obtained when predicting $\myV[\ii]$ from 
$\widehat{\myV}^{(-\ii)}[\ii]=\genm[\ii] + \gencm[\ii,-\ii]\gencm[\ii]^{-1}(\myV[-\ii]-\genm[-\ii])$, where $\myV \sim \mathcal{N}(\genm, \gencm)$.

\begin{theorem}[Fast calculation of Gaussian cross-validation residuals]
	\label{lem1}
Let $\ii \in \sfold$ and $\genQ=\gencm^{-1}$. Then, 
\begin{align}
\label{lem1.1}
e_{\ii}(\myv;\genm,\gencm)=\genQ[\ii]^{-1}(\genQ(\myv-\genm))[\ii]
\ \ \ 
(\myv \in \R^n).
\end{align}	
Consequently, for $\myV \sim \mathcal{N}(\genm, \gencm)$, the $\E_{\ii}=e_{\ii}(\myV;\genm,\gencm)$
$(\ii \in \sfold)$ are jointly Gaussian, centred, and 
\begin{equation}
\label{lem1.2}
\operatorname{Cov}(\E_{\ii},\E_{\jj})
=\genQ[\ii]^{-1} \genQ[\ii,\jj] \genQ[\jj]^{-1} \ \ (\ii,\jj \in \sfold).
\end{equation}
For the case of folds $\ii_1,\dots, \ii_q$ which concatenation gives $(1,\dots,n)$,
$\E_{\ii_1,\dots,\ii_q}:=[\E_{\ii_1}^{\top},\dots,\E_{\ii_q}^{\top}]^{\top}$ has centred $n$-dimensional Gaussian distribution with covariance matrix 
\begin{equation}
\label{lem1.3}
\operatorname{Cov}(\E_{\ii_1,\dots,\ii_q})
=\bloc^{-1} 
\genQ
\bloc^{-1}  
\end{equation}
where $\bloc=\operatorname{blockdiag}\left(\genQ[\ii_1], \dots, \genQ[\ii_q] \right)$.
\end{theorem}

\begin{proof}
	Equation~\ref{lem1.1} follows from suitably expressing blocks of $\gencm$'s inverse using the Schur complement approach of Equations~\ref{horn1} and~\ref{horn2}, delivering respectively  
	\begin{align*}
	\gencm^{-1}[\ii]&=(\gencm[\ii]-\gencm[\ii,-\ii]\gencm[-\ii]^{-1}\gencm[-\ii,\ii])^{-1}\text{, and} \\
	\gencm^{-1}[\ii,-\ii] & =-(\gencm[\ii]-\gencm[\ii,-\ii]\gencm[-\ii]^{-1}\gencm[-\ii,\ii])^{-1}\gencm[\ii,-\ii]\gencm[-\ii]^{-1}.
	\end{align*}
	Considering the rows indexed by $\ii$ of the product $\genQ\myv$, we hence get that 
	\begin{equation*}
	\begin{split}
	(\genQ(\myv-\genm))[\ii]&=\genQ[\ii](\myv[\ii]-\genm[\ii])-\genQ[\ii]\gencm[\ii,-\ii]\genQ[-\ii](\myv[-\ii]-\genm[-\ii])
	\end{split}
	\end{equation*}
	whereof $e_{\ii}(\myv;\genm,\gencm)=(\genQ[\ii])^{-1}(\genQ\myv)[\ii]=\genQ[\ii]^{-1}I_{n}[\ii,]\genQ(\myv-\genm)$. 
	
	The joint Gaussianity of the $\E_\ii$'s and the covariance structure follow from the affine dependence of $\E_\ii=(\genQ[\ii])^{-1}I_{n}[\ii,]\genQ(\myv-\genm)$ on $\myV$, so that the concatenating $q\geq 1$ random vectors $\E_{\ii_1}, \dots, \E_{\ii_q}$ leads to a Gaussian vector by left multiplication of $\myV-\genm$ by a deterministic matrix. As for the covariance matrix of $\E_{\ii_1,\dots,\ii_q}$, denoting by $I_n$ the identity matrix in $\R^{n\times n}$, it follows from 
	\begin{equation*}
	\begin{split}
	[\E_{\ii_1}^{\top},\dots,\E_{\ii_q}^{\top}]^{\top}
	&=\left(
	\begin{matrix} 
	(\genQ[\ii_1])^{-1}I_{n}[\ii_1]\genQ\myV\\
	\dots\\
	(\genQ[\ii_q])^{-1}I_{n}[\ii_q]\genQ\myV
	\end{matrix} 
	\right)
	\\
	&=
	\operatorname{blockdiag}\left( (\genQ[\ii_1])^{-1}, \dots, (\genQ[\ii_q])^{-1} \right)
	\left(
	\begin{matrix} 
	I_{n}[\ii_1]\\
	\dots\\
	I_{n}[\ii_q]
	\end{matrix}
	\right)
	\genQ 
	\myV\\
	&=
	\bloc^{-1} I_n \genQ \myV
	= 
	\bloc^{-1} \genQ \myV,
	\end{split}
	\end{equation*}
	hence
	$
	\operatorname{Cov}(\E_{\ii_1,\dots,\ii_q})
	=\bloc^{-1}
	\genQ
	\gencm
	\genQ
	\bloc^{\top}
	=
	\bloc^{-1}
	\genQ
	\bloc^{-1}
	$.
\end{proof}

\begin{remark}
The first part of the Lemma is purely deterministic, and hence applies to any function approximation methods that boils down to Gaussian conditioning equations (e.g., regularized quadratic risk minimization in Reproducing Kernel Hilbert Spaces). 	
\end{remark}

\begin{remark}
	Note that for arbitary $\ii_1,\dots,\ii_q \in \sfold$ i.e. without imposing ordering between them or that they form a partition, we would have a similar result yet without the above simplification, i.e. $\operatorname{Cov}(\E_{\ii_1,\dots,\ii_q})
	=\bloc
	\Delta  
	\genQ 
	\Delta^{T}  
	\bloc$
	with 
	$\Delta=[I_{n}[\ii_1]^{\top},\dots, I_{n}[\ii_q]^{\top}]^{\top}$. An extreme case would be to consider all possible non-empty subsets of $\{1,\dots,n\}$, leading to $q=2^n-1$ and $n2^{n-1}$ lines for $\Delta$. 
\end{remark}

\begin{remark}
	\label{rk_LOO}
	In the extreme case where $q=n$ and the $\ii_j$'s are set to $(j)$ $(1\leq j \leq n)$, one recovers on the other hand fast leave-one-out cross-validation formulae, and we obtain as a by-product the covariance matrix of leave-one-out errors 
	\begin{equation}
	\label{cov_LOO_res}
	\rm{diag}(\genQ[i]^{-1}) \genQ \rm{diag}(\genQ[i]^{-1}).
	\end{equation}
\end{remark}

\subsection{The Simple Kriging case}
\label{SKcase}

We now come back to GP prediction settings and first focus on Simple Kriging, following the observation model $Z_{i}=\gp(\xx_i)+\varepsilon_i \ \ \ (\xx_i \in D, i \in \{1,\dots,n\})$ from the previous section, where we recall that $\gp(\xx)=\mu(\xx)+\cgp(\xx)$ ($\xx\in D$) where $\cgp$ is a centred GP with covariance kernel $k$ and $\mu$ is a trend function, and $\noise\sim \mathcal{N}(\mathbf{0}, \covnoise)$ independently of $\cgp$. $\mu$ is first assumed null, an assumption that we will relax later on. Like in the exposition of the previous section, we denote $\Z$ the observation vector and, for any $\ii \in \sfold$, and $\E_{\ii}=\Z[\ii]-\cvpred{\ii}$ where $\cvpred{\ii}$ is the BLUP of $\Z[\ii]$ based on $\Z[-\ii]$. 
We assume throughout that $\covsum=K+\covnoise$ is invertible, with $K=(k(\xx_i,\xx_j))_{i,j\in\{1,\dots,n\}}$. 

\begin{cor}{}
\label{thm1}
The $\E_{\ii}$ $(\ii \in \sfold)$ are jointly Gaussian, centred, and 
\begin{equation}
\label{thm1.2}
\operatorname{Cov}(\E_{\ii},\E_{\jj})
=\Q[\ii]^{-1} \Q[\ii,\jj] \Q[\jj]^{-1} \ \ (\ii,\jj \in \sfold),
\end{equation}
where $\Q=\covsum^{-1}$. 
For the case of folds $\ii_1,\dots, \ii_q$ which concatenation gives $(1,\dots,n)$,
$\E_{\ii_1,\dots,\ii_q}:=[\E_{\ii_1}^{\top},\dots,\E_{\ii_q}^{\top}]^{\top}$ has centred $n$-dimensional Gaussian distribution with covariance matrix 
\begin{equation}
\label{thm1.3}
\operatorname{Cov}(\E_{\ii_1,\dots,\ii_q})
=\bloc^{-1} 
Q
\bloc^{-1}  
\end{equation}
where $\bloc=\operatorname{blockdiag}\left(Q[\ii_1], \dots, Q[\ii_q] \right)$.
\end{cor}

\begin{proof}
Directly follows from Lemma~\ref{lem1} using 
$\E_{\ii}=e_{\ii}(\Z;\mathbf{0},\covsum)$ with $\covsum=K+\covnoise$. 
\end{proof}

\begin{remark}
Extending these results to the case of Simple Kriging with a known trend function $\mu$ directly follows from $\E_{\ii}=e_{\ii}(\Z;\mymu,\covsum)$ with $\mymu=(\mu(\xx_i))_{1\leq i \leq n}$. 
\end{remark}

\begin{remark}
\label{rk_LOO2}
In the case of Remark~\ref{rk_LOO2}, we then obtain  
$\operatorname{cov}(E_i,E_j)=\Q[i]^{-1} \Q[i,j] \Q[j]^{-1}$ whereof $\operatorname{corr}(E_i,E_j)=Q[i]^{-\frac{1}{2}} \Q[i,j] \Q[j]^{-\frac{1}{2}}$, leading in particular to a correlation matrix of LOO residuals
\begin{equation}
\label{cor_LOO_res}
\operatorname{Corr}(\E_{(1),\dots,(n)})
=
\rm{diag}(\Q[i]^{-\frac{1}{2}}) \Q \rm{diag}(\Q[i]^{-\frac{1}{2}})=\bloc^{\frac{1}{2}}\Q\bloc^{\frac{1}{2}} 
.
\end{equation} 
\end{remark}

As a by-product of Propostion~\ref{thm1} (in the limiting case of LOO, see Remark~\ref{rk_LOO}), we represent on Figure~\ref{xiong_CorrLOO} the correlations between the LOO prediction residual at the leftmost location versus the LOO residuals at the other locations in the settings of the opening example from Figure~\ref{xiong_preds}. 

\medskip 

\begin{figure}[h!]
	\begin{center}
		\includegraphics[width=.65\linewidth]{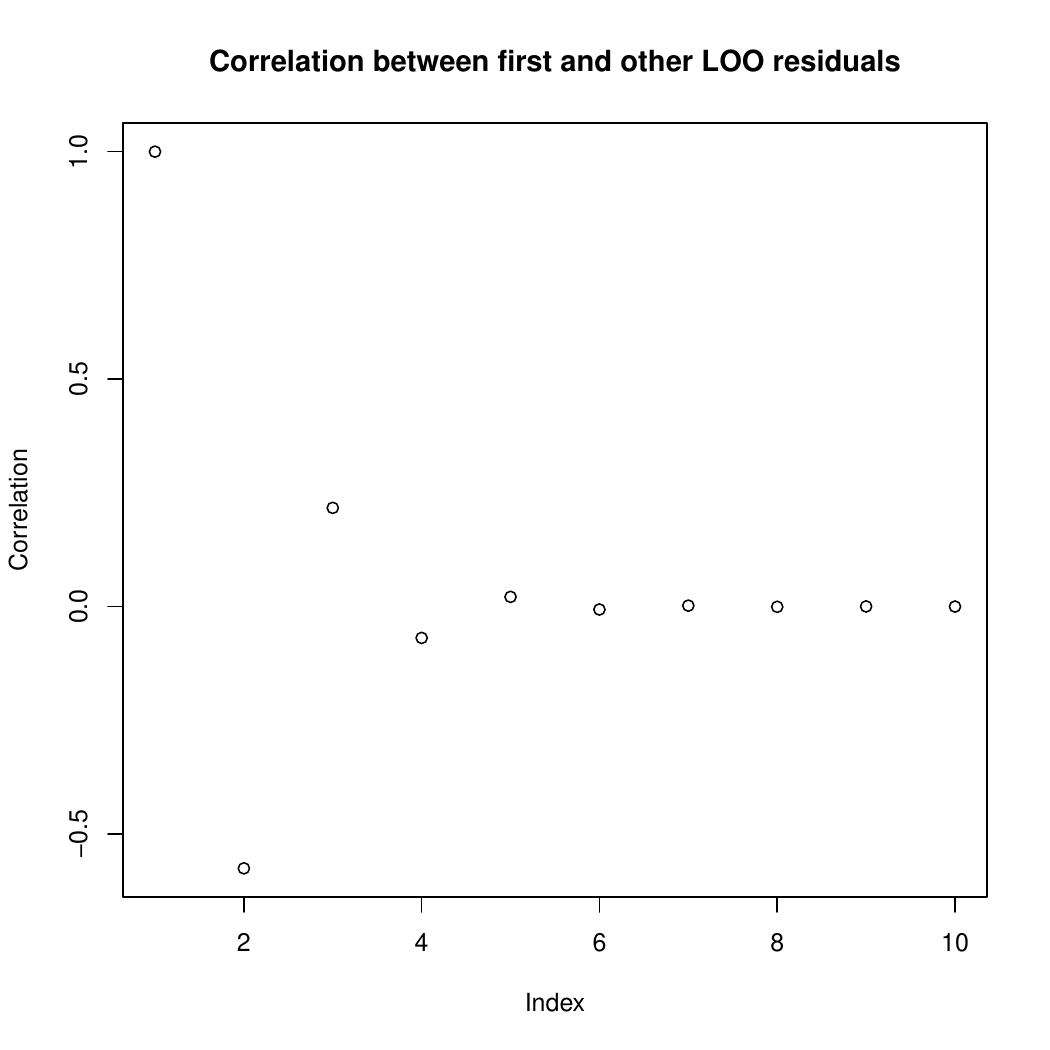}
		\caption{Correlations between the LOO residual at the leftmost point and the LOO residuals at all design points in the example displayed on Figure~\ref{xiong_preds}.}
		\label{xiong_CorrLOO}
	\end{center}
\end{figure}

A first interesting thing that we would like to stress here is that correlation with the LOO prediction residual at the second location is negative, with a value below $-0.5$. While such a negative correlation between successive LOO residuals does not materialize with the actual prediction errors at these two locations, an illustration of this effect can be seen on the upper panel of Figure~\ref{xiong_preds} when looking at the fourth and fifth locations. Coming back to Figure~\ref{xiong_CorrLOO}, after the second location the correlation goes back to positive but with a smaller magnitude and subsequently appear to continue with a damped oscillation until stationing around zero. As we will see in Section~3, accounting for these correlations is instrumental in producing suitable QQ-plots from cross-validation residuals. Let us now turn to a further example focusing on multiple-fold cross-validation.

\subsection{Extension to Universal Kriging case and more}
\label{Extensions}

Let us now formulate a corollary to Theorem~\ref{thm1} for the case of Universal Kriging such as presented throughout Section~\ref{secbasics} in general and transductive settings. 

\begin{cor}{}
	\label{cor1}
	For any $\ii \in \sfold$ such that $F[-\ii ,]$ is full column-ranked using the notation $\widehat{\Z}^{(-\ii)}[\ii]$ to denote the BLUP of $\Z[\ii]$ based on $\Z[-\ii]$ in Universal Kriging settings, the residual $\E_{\ii}=\Z[\ii]-\widehat{\Z}^{(-\ii)}[\ii]$ writes
	\begin{equation}
	\begin{split}
	\label{MFCVres_UK2}
	\E_{\ii}&=(\augcm^{-1}[\ii])^{-1}(\augcm^{-1}[-\ii]\Z)[\ii],\\
	&=(\Qt[\ii])^{-1}(\Qt\Z)[\ii],
	\end{split}
\end{equation} 
	where $\augcm$ is defined in Eq.~\ref{defM}, and $\Qt=\Q-\Q F (F^{\top}\Q F)^{-1} F^{\top} \Q$ with $\Q=\covsum^{-1}=\covsum^{-1}$. 
	Consequently, for any $q>1$ and $\ii_1,\dots, \ii_q \in \sfold$ such that the $F[-\ii_{j} ,]$'s are full column-ranked, the $\E_{\ii_j}$ $(1\leq j \leq q)$ are jointly Gaussian, centred, and with covariance structure given by 
	\begin{equation}
	\operatorname{Cov}(\E_{\ii},\E_{\jj})
	=(\Qt[\ii])^{-1} \Qt[\ii,\jj] (\Qt[\jj])^{-1} \ \ (\ii,\jj \in \{\ii_1,\dots, \ii_q\}).
	\end{equation}
	In particular, for the case of index vectors $\ii_1,\dots, \ii_q$ forming a partition of $\{1,\dots,n\}$
	and such that the concatenation of $\ii_1,\dots, \ii_q$ gives $(1,\dots,n)$,
	then the concatenated vector of cross-validation residuals $\E_{\ii_1,\dots,\ii_q}:=[\E_{\ii_1}^{\top},\dots,\E_{\ii_q}^{\top}]^{\top}$ has a $n$-dimensional centred Gaussian distribution with covariance matrix 
	\begin{equation}
	\label{thm1.3}
	\operatorname{Cov}(\E_{\ii_1,\dots,\ii_q})
	=\widetilde{\bloc}
	\Qt
	\widetilde{\bloc},  
	\end{equation}
	where $\widetilde{\bloc}=\operatorname{blockdiag}\left( (\Qt[\ii_1])^{-1}, \dots, (\Qt[\ii_q])^{-1} \right)$.
\end{cor}

\begin{proof}
A direct way to prove this property is to first generalize Equation~\ref{MFCVres_UK} to the case of generic folds $\ii \in \sfold$ such that $F[-\ii ,]$ is full column-ranked, which establishes Equation~\ref{MFCVres_UK2}, and then to follow the steps of Proposition~\ref{lem1} regarding the covariance structure and the joint Gaussian distribution of cross-validation residuals. Another proof takîng a Bayesian route is sketched below. 
\end{proof}

\begin{remark}
As we have recalled in Section~2, it is known that Universal Kriging equations can  be retrieved in a Bayesian setting by endowing $\mybeta$ with a Gaussian prior distribution $\mathcal{N}(\mathbf{0}, \covbeta)$ with invertible $\covbeta$ and letting $\covbeta^{-1}$ tend to $\mathbf{0} \in \R^{p \times p}$. Under this model (at fixed invertible $\covbeta$), the random vector $\Z$ would be centred Gaussian with a covariance matrix $\covsum+F \covbeta F^{\top}$.  
Then one may use Proposition~\ref{lem1} with $\mathbf{m}=0$ and $P=\covsum+F \covbeta F^{\top}$ to find that the cross-validation errors $\Z[\ii]-\widehat{\Z}^{(-\ii)}_{\covbeta}[\ii]$ are jointly Gaussian with covariance structure $$\operatorname{Cov}(\Z[\ii]-\widehat{\Z}^{(-\ii)}_{\covbeta}[\ii], \Z[\jj]-\widehat{\Z}^{(-\jj)}_{\covbeta}[\ii])
=
\genQ_{\covbeta}[\ii]^{-1} \genQ_{\covbeta}[\ii,\jj] \genQ_{\covbeta}[\jj]^{-1}
$$ 
where 
$\widehat{\Z}^{(-\ii)}_{\covbeta}[\ii]$ denotes the BLUP of $\Z[\ii]$ based on $\Z[-\ii]$ under the prior distribution $\mathcal{N}(\mathbf{0}, \covbeta)$ on $\covbeta$, and $\genQ_{\covbeta}=(\covsum+F \covbeta F^{\top})^{-1}$. Now by block inversion and Woodbury formula, we get
\begin{equation*}
\begin{split}
\genQ_{\covbeta}&=\left( \begin{matrix} \covsum  & F \\ F^{\top}  &  -\covbeta^{-1} \end{matrix} \right)^{-1}[1:n]\\
&=\Q-\Q F (\covbeta^{-1}+F^{\top}\Q F)^{-1} F^{\top} \Q.
\end{split}
\end{equation*}
From there we can conclude since $\covbeta^{-1}\to \mathbf{0}$, $P_{\covbeta} \to \Qt$. Note that checking tightness is not necessary here as we are considering a finite number of centred Gaussian residuals with converging covariances. 
\end{remark}

\begin{remark}
The latter corollary enables retrieving fast formulas for multiple-fold cross-validation residuals of linear regression models of \cite{Zhang1993} as well as associated covariances. In fact, putting $K=\mathbf{0}$ and $\covnoise=\tau^2 I_n$ (where $\tau>0$) delivers $\tau^2 \Qt=I_n-H$ with $H=(F^{\top}F)^{-1}F^{\top}\Z$, whereof $\E_{\ii}=(I_{n}-H)[\ii]^{-1} (\Z-H\Z)[\ii] $ and 
\begin{equation*}
\operatorname{Cov}(\E_{\ii},\E_{\jj})
=(I_{n}-H)[\ii]^{-1} (I_{n}-H)^{-1}[\ii,\jj] (I_{n}-H_\lambda)[\jj]^{-1} \ \ (\ii,\jj \in \sfold).
\end{equation*}	
Similarly, in ridge regression, we obtain by putting $\covnoise=\tau^2 I_n$ and $K=\lambda I_n$, 
\begin{equation}
\E_{\ii}=(I_{n}-H_\lambda)[\ii]^{-1} (\Z-H_{\lambda}\Z)[\ii].
\end{equation} 
where $H_\lambda=F(F^{T}F+\lambda I_n)^{-1}F^{T}$. 
Consequently, for any $q>1$ and $\ii_1,\dots, \ii_q \in \sfold$, 
the $\E_{\ii_j, \lambda}$ $(1\leq j \leq q)$ are jointly Gaussian, centred, and with covariance structure given by 
\begin{equation}
\operatorname{Cov}(\E_{\ii},\E_{\jj})
=(I_{n}-H_\lambda)[\ii]^{-1} (I_{n}-H_\lambda)^{-1}[\ii,\jj] (I_{n}-H_\lambda)[\jj]^{-1} \ \ (\ii,\jj \in \sfold).
\end{equation}
\end{remark}

Many other models could enjoy similar results (encompassing kernel ridge regression and least squares support vector machines, see \cite{An2007} where similar fast computations have been obtained for these model classes, yet without investigating the covariances and joint distribution of CV residuals) but in the following we will stick to GP as it is our main model of interest and the principles are readily applicable to other predictors falling under the settings of Proposition~\ref{lem1}.

\section{Some consequences in GP model fitting}
\label{csq_fitting}

\subsection{Graphical diagnostics and pivotal statistics}
\label{diag}

We now wish to further explore the potential of the established formulae for diagnosing the quality of GP models.  
A first immediate consequence of Theorem~\ref{thm1} and Corollary~\ref{cor1} for model fitting is that, under the hypothesis that observations are generated from the assumed model, cross-validation residuals can be transformed via appropriate matrix multiplications into a random vector with standard multivariate normal distribution.
Assuming for simplicity that $\operatorname{Cov}(\E_{\ii_1,\dots,\ii_q})$ is full-ranked  
and denoting $t=\sum_{j=1}^q \# \ii_j$ it follows that  
\begin{equation*}
\operatorname{Cov}(\E_{\ii_1,\dots,\ii_q})^{-1/2}\E_{\ii_1,\dots,\ii_q} \sim 
\mathcal{N}\left( \mathbf{0}, I_{t} \right). 
\end{equation*}
If fact, any matrix $S \in \R^{n\times n}$ such that $S\bloc \covsum^{-1}\bloc S^{\top}= I_{n}$ (remaining for simplicity under the final assumptions of Theorem~\ref{thm1}) does the job. More specifically, sticking to the final part of Theorem~\ref{thm1} where $\Delta=I_{n}$, one gets indeed with $S=\covsum^{1/2}\bloc^{-1}$ 
\begin{equation*}
S\E_{\ii_1,\dots,\ii_q} = \covsum^{-1/2} \Z \sim 
\mathcal{N}\left( \mathbf{0}, I_{n} \right), 
\end{equation*}
and hence the hypothesis that the model is correct can be questioned using any standard means relying on such a pivotal quantity with multivariate Gaussian distribution (e.g., by means of a $\chi^2$ test statistic or graphical diagnostics such as Q-Q plots). 
Similar considerations apply to the UK settings of Corollary~\ref{cor1}.

\begin{remark}
If $\Delta$ is not $I_n$ but still an invertible matrix, 
then the above can be adapted with $S=\covsum^{1/2}\Delta^{-1}\bloc^{-1}$. In case $\Delta$ is not invertible, analoguous approaches can be employed relying on the Moore-Penrose pseudo-inverse of $\bloc\Delta \covsum^{-1} \Delta \bloc^{\top}$ where the $n$-dimensional standard normal should be replaced by a standard Gaussian on the range of $\Delta \covsum^{-1} \Delta$. 
\end{remark}

Coming back to our first one-dimensional example, Figure~\ref{xiong_QQplotsLOO} represents Q-Q plots of the distribution of standardized LOO residuals (right panel) versus ``transformed'' LOO residuals (left panel). 

\begin{figure}[h!]
	\includegraphics[width=.5\linewidth]{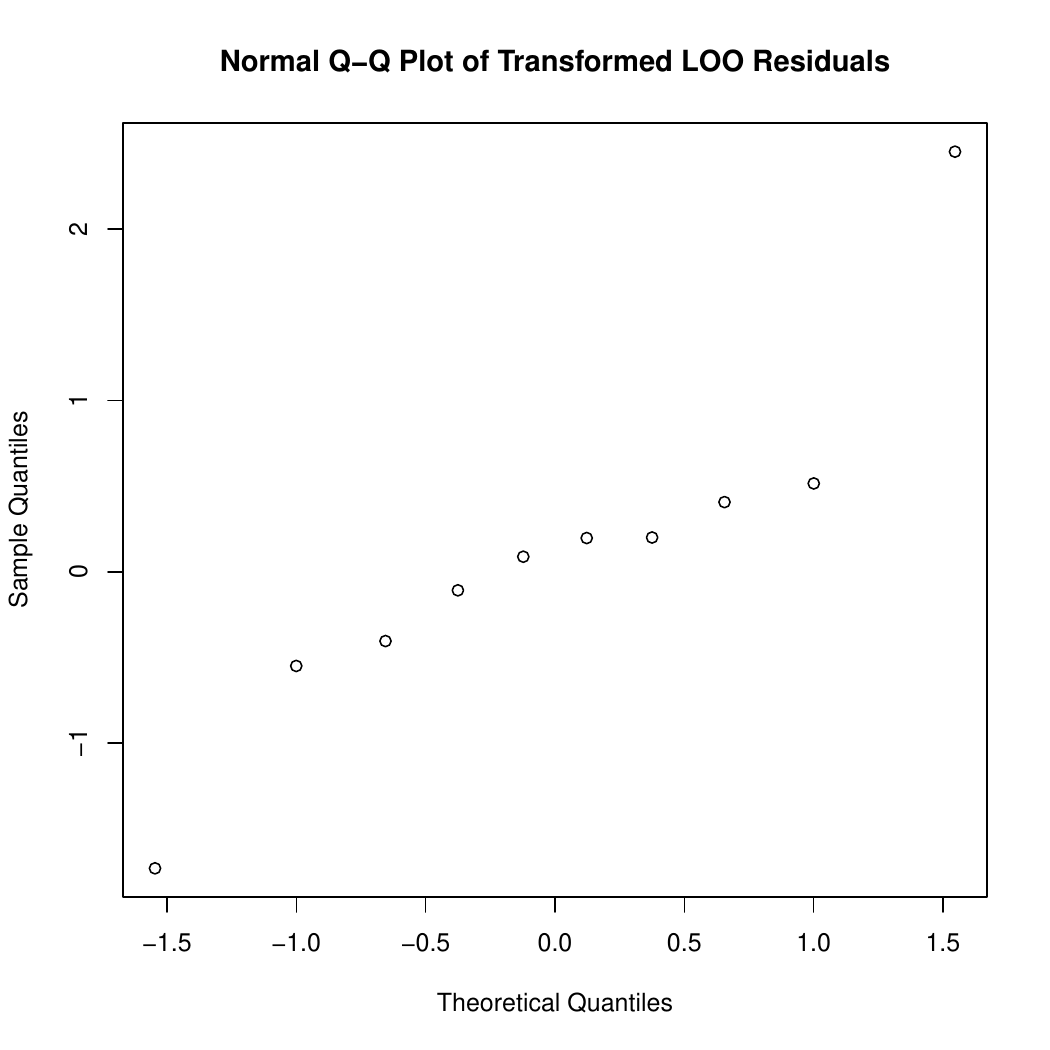}
	\includegraphics[width=.5\linewidth]{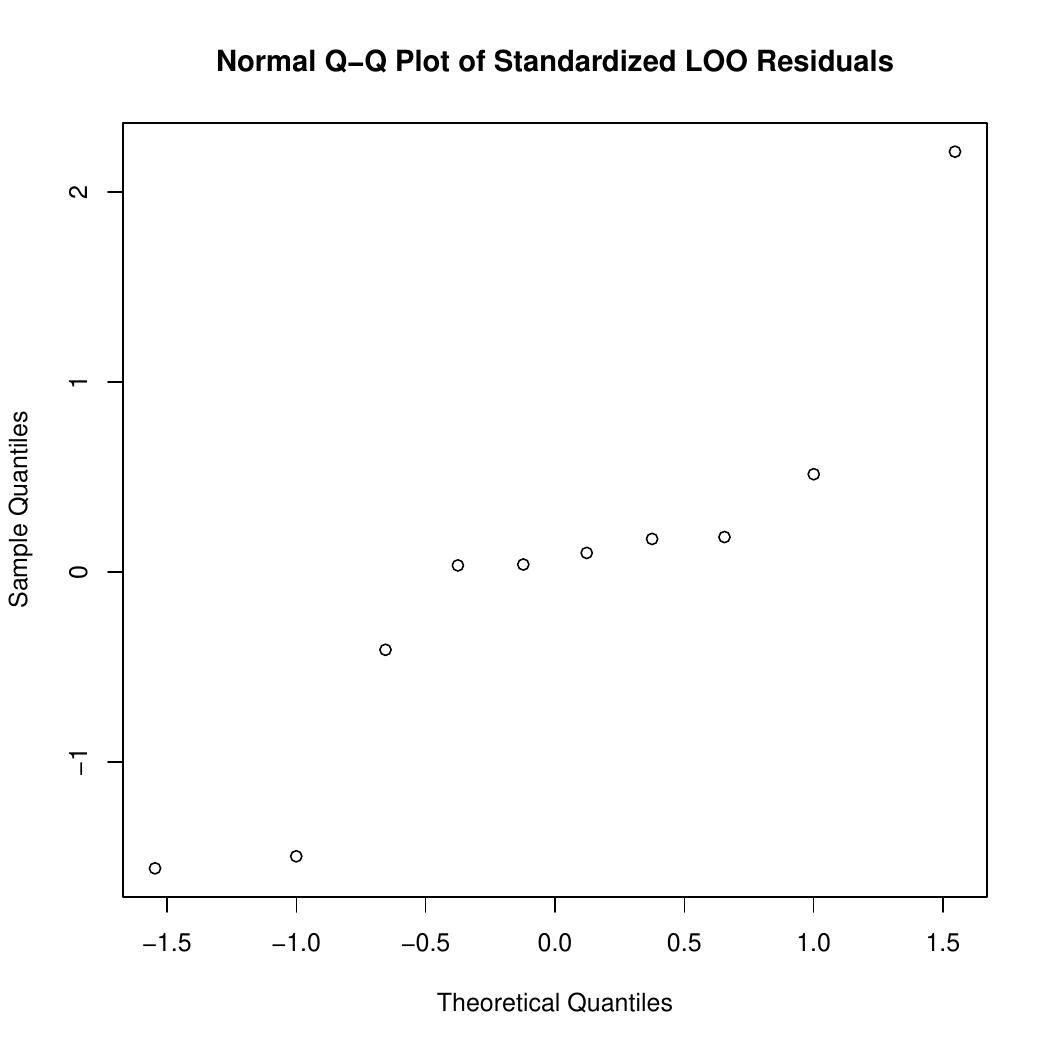}	
\caption{On the effect of accounting and correcting for correlation in Q-Q plots based on LOO residuals. Right panel: Q-Q plot against $\mathcal{N}(0,1)$ of LOO residuals merely divided by corresponding LOO standard deivations. Left panel: Q-Q plot against $\mathcal{N}(0,1)$ of duly transformed 
	LOO residuals.}
	\label{xiong_QQplotsLOO}
\end{figure}

\bigskip

While these two plots possess quite similar appearances, a difference that stands out occurs at the second point from the left. Even if it does not really make sense to think of points in terms of unique locations for the left panel of of Figure~\ref{xiong_QQplotsLOO}, we hypothesize that most entries do not differ much from the standardized version except for those cases where two successive locations come with LOO residuals of relatively high magnitude, which is precisely the case for the fourth and fifth locations from Figure~\ref{xiong_QQplotsLOO}, respectively with large negative and positive LOO residuals. 

As it turns out,  the second largest negative LOO residuals from the right panel of Figure~\ref{xiong_QQplotsLOO} does not show up on the left panel, and it seems reasonable to attribute this to its decorrelation from the largest positive residual corresponding to the fifth location (appearing rightmost of the right panel). 

\bigskip 

We now turn to the topic of covariance parameter estimation relying on cross-validation residuals. As we will see below, incorporating covariances between CV cross-validation residuals within such approaches leads to interesting developments. 

\subsection{Some implications in scale parameter estimation}
\label{scale}

Assuming a kernel of the form $k(\xx,\xx')=\sigma^2 r(\xx,\xx')$ with unknown $\sigma^2$ and noiseless observations (which is commonly encountered in the realm of computer experiments), the problem of estimating $\sigma^2$ by maximum likelihood versus cross-validation has been considered in earlier works, leading notably to results concerning the superiority of MLE in the well-specified case but also to situations when CV turned out to perfom better than MLE in the misspecified case \cite{Bachoc2013}. Denoting by $R$ the correlation matrix associated with the observation points $\xx_1,\dots,\xx_n$, assumed full-rank, it is well-known (Cf. for instance \cite{Santner.etal2003}) that the MLE of $\sigma^2$ can be written in closed form as 
\begin{equation}
\widehat{\sigma}^2_{\text{ML}}=\frac{1}{n} \Z R^{-1} \Z,
\end{equation}
In contrast, the leave-one-out-based estimator of $\sigma^2$ investigated in \cite{Bachoc2013} reads 
\begin{equation}
\label{def_Bachoc_est}
\widehat{\sigma}^2_{\text{LOO}}=\frac{1}{n} \Z R^{-1} (\operatorname{diag}(R^{-1}))^{-1} R^{-1} \Z,
\end{equation}
and originates from the idea, traced back by \cite{Bachoc2013} to \cite{Cressie1993}, that the criterion (with notation $C_{\text{LOO}}$ inherited from \cite{Bachoc2013}) defined by 
\begin{equation}
C_{\text{LOO}}^{(1)}(\sigma^2)=\frac{1}{n} \sum_{i=1}^n \frac{\E_{i}^2}{\sigma^2 c_{-i}^2},
\end{equation}   
should take a value close to one, where 
$c_{-i}^2=(s_{-i}^2)/\sigma^2$, leading to $\widehat{\sigma}^2_{\text{LOO}}=\frac{1}{n} \sum_{i=1}^n \frac{\E_{i}^2}{c_{-i}^2}$ and ultimately to Eq. \ref{def_Bachoc_est}.
This estimator turns out to be unbiased as $\widehat{\sigma}^2_{\text{LOO}}$, yet with greater variance (see \cite{Bachoc2013}): 
\begin{equation}
\operatorname{var}\left(\widehat{\sigma}^2_{\text{LOO}}\right)=\frac{2\sigma^4 \text{tr}((R^{-1} (\operatorname{diag}(R^{-1}))^{-1})^2)}{n^2}
\geq
\frac{2\sigma^4}{n}
=\operatorname{var}\left(\widehat{\sigma}^2_{\text{ML}}\right). 
\end{equation} 
In light of former considerations on the joint distribution of cross-validation residuals, it appears natural to revise $C^{(1)}_{\text{LOO}}(\sigma^2)$ in order to correct for covariances between LOO residuals, resulting in   
\begin{align}
C_{\widetilde{\text{LOO}}}^{(1)}(\sigma^2)
&=\frac{1}{n} \sum_{i=1}^n \sum_{j=1}^n \E_i (\bloc K^{-1} \bloc)^{ij} \E_j
=\frac{1}{n \sigma^2} \E^{\top} \operatorname{diag}(R^{-1}) R \operatorname{diag}(R^{-1}) \E
=\frac{1}{n \sigma^2} \Z^{\top} R^{-1} \Z,
\end{align}  
so that setting this modified criterion to $1$ would plainly result in 
$\widehat{\sigma}^2_{\widetilde{\text{LOO}}}=\frac{1}{n} 
\Z^{\top} R^{-1} \Z=\widehat{\sigma}^2_{\text{ML}}$
Hence the advantages found for $\widehat{\sigma}^2_{\text{LOO}}$ in \cite{Bachoc2013} in the case of model misspecification come at the price of neglecting covariances and hence possible redundancies between leave-one-out residuals, yet attempting to address this issue by naturally accounting for those covariances within the criterion leads to the MLE, known to enjoy optimality properties in the well-specified case of the considered settings but to be potentially suboptimal otherwise. We will see next that this is not the only instance where modifying LOO in order to account for residual covariances leads back to MLE.   

\subsection{On the estimation of further covariance parameters}
\label{further}

Cross-validation has also been used to estimate covariance parameters beyond the scale parameter $\sigma^2$.  We will now denote by $\theta$ those additional parameters (and by $\psi=(\sigma^2, \theta)$ the concatenation of all covariance parameters) and review some approaches that have been used to estimate them via cross-validation. An important point compared to the last section on $\sigma^2$'s estimation is that Kriging predictors and CV residuals generally depend on $\theta$, so that we will now stress it by noting $\E(\theta)$. Still in noiseless settings, the approach followed by \cite{Santner.etal2003,Bachoc2013} to estimate $\theta$ based on leave-one-out cross-validation residuals is to minimize 
\begin{equation}
\label{CLOO}
C_{\text{LOO}}^{(2)}(\theta)
=\sum_{i=1}^n \E_i(\theta)^2
,
\end{equation} 
Following \cite{Bachoc2013}, the latter criterion can be in turn expressed in closed form as 
\begin{equation}
\label{CLOO_closed}
C_{\text{LOO}}^{(2)}(\theta)
=\Z^{\top}R_{\theta}^{-1} \operatorname{diag}(R_{\theta}^{-1})^{-2}R_{\theta}^{-1}\Z,
\end{equation}  
where $R_{\theta}$ stands for $\Z$'s correlation matrix under correlation parameter $\theta$.
Based on our results, a natural extension of this criterion to multiple fold cross-validation can be similarly obtained. Considering $q$-fold settings such as in the final part of Theorem~\ref{thm1}, $C_{\text{LOO}}^{(2)}$ becomes indeed
\begin{equation}
\label{C2CV}
C_{\text{CV}}^{(2)}(\theta)
=\sum_{j=1}^q \lvert\lvert\E_{\ii_j}(\theta)\rvert\rvert^2
=
\lvert\lvert\E_{\ii_1,\dots,\ii_q}(\theta)\rvert\rvert^2
.
\end{equation} 
Then, building up upon our main theorem, we obtain that 
\begin{equation}
\label{C2CV_closed}
C_{\text{CV}}^{(2)}(\theta)=\mathbf{Z}^{\top} R_{\theta}^{-1} \bar{\bloc}^2(\theta) R_{\theta}^{-1} \mathbf{Z},
\end{equation}
with $\bar{\bloc}(\theta)=\operatorname{blockdiag}\left( (R_{\theta}^{-1}[\ii_1])^{-1}, \dots, (R_{\theta}^{-1}[\ii_q])^{-1} \right)$, generalizing indeed Eq.\ref{CLOO_closed}. 
Note that applying a sphering transformation to the CV residuals previous to taking norms in Eqs.~\ref{CLOO},\ref{C2CV} would lead to a criterion that boils down to $\theta \mapsto \mathbf{Z}^{T}R_{\theta}^{-1}\mathbf{Z}$ and hence appears as one of the two building blocks of the log-likelihood criterion. 

\bigskip 

On a different note and relaxing now the assumption of noiseless observations, Eq.~\ref{CLOO} is only one among several possible ways to construct a criterion based on applying loss functions to leave-one-out residuals, as exposed in \cite{Rasmussen.Williams2006}. While Eq.~\ref{CLOO} corresponds to the squared error loss, one arrives for instance by using instead the ``negative log validation density loss'' (after the terminology of \cite{Rasmussen.Williams2006}, Section 5.4.2) at the ``LOO log predictive probability'', also sometimes called \textit{pseudo-likelihood}: 
\begin{equation}
C_{\text{LOO}}^{(3)}(\psi)
=\sum_{j=1}^n \log\left( 
p_{Z[j] | \mathbf{Z}[-j]}(z_j| \mathbf{z}[-j]; \psi)
\right), 
\end{equation} 
where $p_{Z[j] | \mathbf{Z}[-j]}(z[j]| \mathbf{z}[-j]; \psi) $ denotes the conditional density of $Z[j]$ at the value $z[j]$ knowing that $\mathbf{Z}[-j]=\mathbf{z}[-j]$. 
Further scoring rules could be used as a base to define further criteria relying on cross-validation residuals, as illustrated for instance with the CRPS score in \cite{Petit, Petit_phd}. 

\bigskip

Let us now focus on an extension of $C_{\text{LOO}}^{(3)}(\psi)$ to multiple fold cross-validation, and on the exploration of some links between the resulting class of criteria and the log-likelihood. First,  $C_{\text{LOO}}^{(3)}(\psi)$ is straightforwardly adapted into
\begin{equation}
C_{\text{CV}}^{(3)}(\psi)
=\sum_{j=1}^q \log\left( 
p_{\mathbf{Z}[\ii_j] | \mathbf{Z}[-\ii_j]}(\mathbf{z}[\ii_j]| \mathbf{z}[-\ii_j]; \psi)
\right).
\end{equation} 
Reformulating the sum of log terms as the log of a product, we see that
\begin{equation}
\label{CV3_prod}
\begin{split}
C_{\text{CV}}^{(3)}(\psi)
=\log\left(\prod_{j=1}^q  
p_{\mathbf{Z}[\ii_j] | \mathbf{Z}[-\ii_j]}(\mathbf{z}[\ii_j]| \mathbf{z}[-\ii_j]; \psi)\right)
=\log\left(\prod_{j=1}^q  
p_{\mathbf{E}_{\ii_j}}(\mathbf{e}_{\ii_j}; \psi)\right)
\end{split}
\end{equation} 
and so it appears that this approach amounts indeed in some sense to abusingly assuming independence betwen the cross-validation residuals $\mathbf{E}_{\ii_{j}}$ $(j=1\dots q)$. 

\bigskip 

We now examine in more detail (still in the settings of the final part of Theorem~\ref{thm1}) when such cross-validation residuals are independent and show that in such case $C_{\text{CV}}^{(3)}(\psi)$ coincides with the log-likelihood criterion based on $\Z$. 
 
\begin{cor}{}
	\label{cor2}
The following propositions are equivalent:
\begin{description}
	\item[a)] The cross-validation residuals $\mathbf{E}_{\ii_{j}}$ $(j=1\dots q)$ are mutually independent, 
	\item[b)] $\bloc=\Q$,
	\item[c)] the subvectors  $\mathbf{Z}[\ii_{j}]$ $(j=1\dots q)$ are mutually independent. 
\end{description}
If they hold, then $C_{\text{CV}}^{(3)}(\psi)=\log(p_{\Z}(\mathbf{z};\psi))$ $(\psi \in \Psi)$. 
\end{cor}
\begin{proof}
As we know from Theorem~\ref{thm1} that the cross-validation residuals $\mathbf{E}_{\ii_{j}}$ $(j=1\dots q)$ form a Gaussian Vector, they are mutually independent if and only if their cross-covariance matrices are null. Yet we also know by Eq.~\ref{thm1.3} 
from the same theorem that $\operatorname{Cov}(\E_{\ii},\E_{\jj})
	=(\Q[\ii])^{-1} \Q[\ii,\jj] (\Q[\jj])^{-1} \ \ (\ii,\jj \in \sfold)$. 
It follows that \textbf{a)} is equivalent to $\Q[\ii_{j},\ii_{j'}]=\mathbf{0}$ for $j \neq j'$ ($j,j'\in\{1,\dots,q\}$), which coincides in turn with \textbf{b)}. Looking finally at \textbf{b)} from the perspective of $\Sigma=\Q^{-1}$ being block-diagonal, we obtain similarly that \textbf{b)} is equivalent to \textbf{c)}. If these propositions hold, then $\mathbf{E}_{\ii_{j}}=\Z[\ii_{j}]$ $(j=1\dots q)$ from Eq.~\ref{lem1.1}, and thus for 
$\psi \in \Psi$:
\begin{equation*}
\label{CV3_indep}
\begin{split}
C_{\text{CV}}^{(3)}(\psi)
=\log\left(\prod_{j=1}^q  
p_{\mathbf{E}_{\ii_j}}(\mathbf{e}_{\ii_j}; \psi)\right)
=\log\left(\prod_{j=1}^q  
p_{\mathbf{Z}[\ii_j]}(\mathbf{z}[\ii_j]; \psi)\right)
=\log(p_{\Z}(\mathbf{z};\psi)).
\end{split}
\end{equation*} 
Also, one could have derived it from $p_{\mathbf{Z}[\ii_j] | \mathbf{Z}[-\ii_j]}(\mathbf{z}[\ii_j]| \mathbf{z}[-\ii_j]; \psi)=p_{\mathbf{Z}_{\ii_j}}(\mathbf{z}[\ii_j]; \psi)$.
\end{proof}

\section{Computational aspects}
\label{comp}

\subsection{Accuracy and speed-up checks for fast formulae}
\label{speedup}

We have seen in previous sections that block matrix inversion enabled us to reformulate the equations of (multiple-fold) cross-validation residuals and associated covariances in compact form, from Simple Kriging to Universal Kriging settings. 
Here and in Appendix we briefly illustrate that our Schur-complement-based formuale do indeed reproduce results obtained by a more straighforward approach, also keeping an eye on computational speed-ups and by-products. The next section focuses in more detail on the respective computational complexities of the two approaches, and also how efficiency gains evolve depending on the number of folds. 

\bigskip

\begin{figure}[h!]
	\begin{center}
		\includegraphics[width=.65\linewidth]{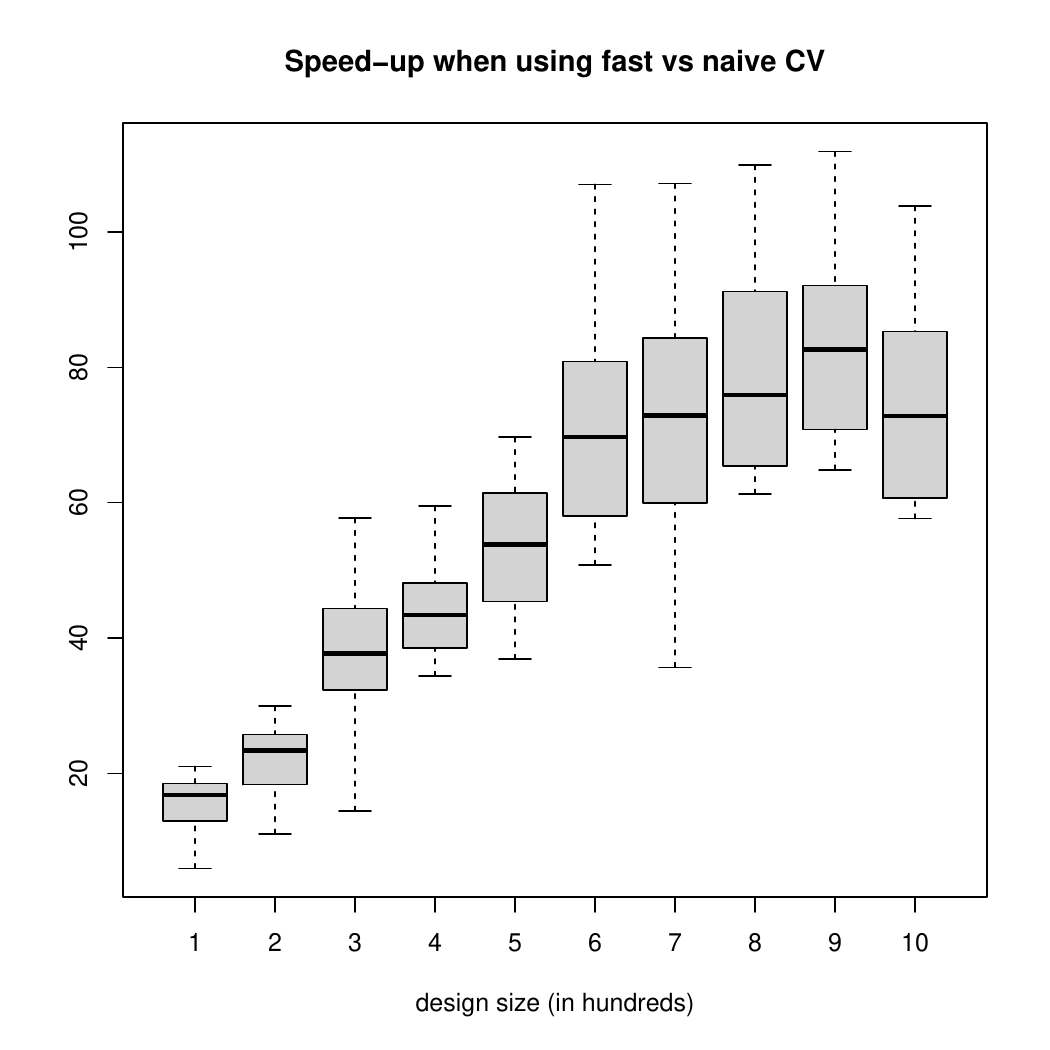}
\caption{Speed-up (ratio between times required to run the naive and fast methods) measured for $50$-fold CV on $10$ regular designs, with $100$ to $1000$ points equidistributed on $[0,1]$, where each measure is repeated $50$ times.}
	\label{speedupCV}
	\end{center}
\end{figure}

Accuracy and speed-up checks for LOO are presented in Section~\ref{LOO_speedups}. We present here some numerical results multiple-fold cross-validation beyond leave-one-out. We compare on Figures~\ref{speedupCV} and  \ref{speedupCV_accu} results obtained via fast versus naive implementation in the case of $50$-fold cross-validation on the same regular designs as in Section~\ref{LOO_speedups} ($10$ regular designs, with $100$ to $1000$ points equidistributed on $[0,1]$). Figure~\ref{speedupCV_accu} represents relative errors on CV mean and covariances compared to a straightforward computational approach, illustrating a very good agreement (with relative errors having orders of magnitude between $10^{-14}$ and $10^{-11}$). As can be seen on Figure~\ref{speedupCV}, speed-ups start getting mitigated with designs nearing the highest considered design sizes, a phenomenon related to a computational trade-off that is analyzed in more detail in the next section.

\begin{figure}[h!]
	\begin{center}
		\includegraphics[width=.45\linewidth]{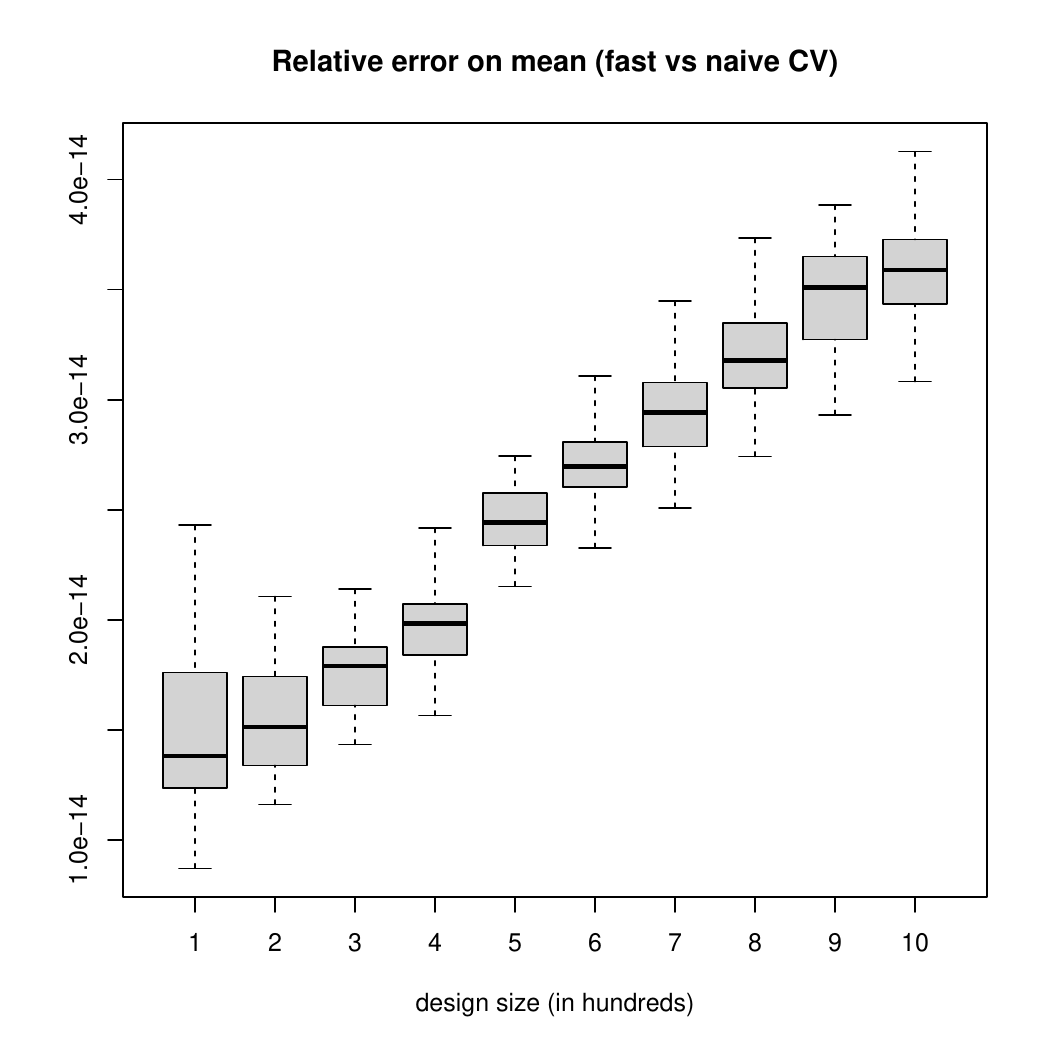}	
		\includegraphics[width=.45\linewidth]{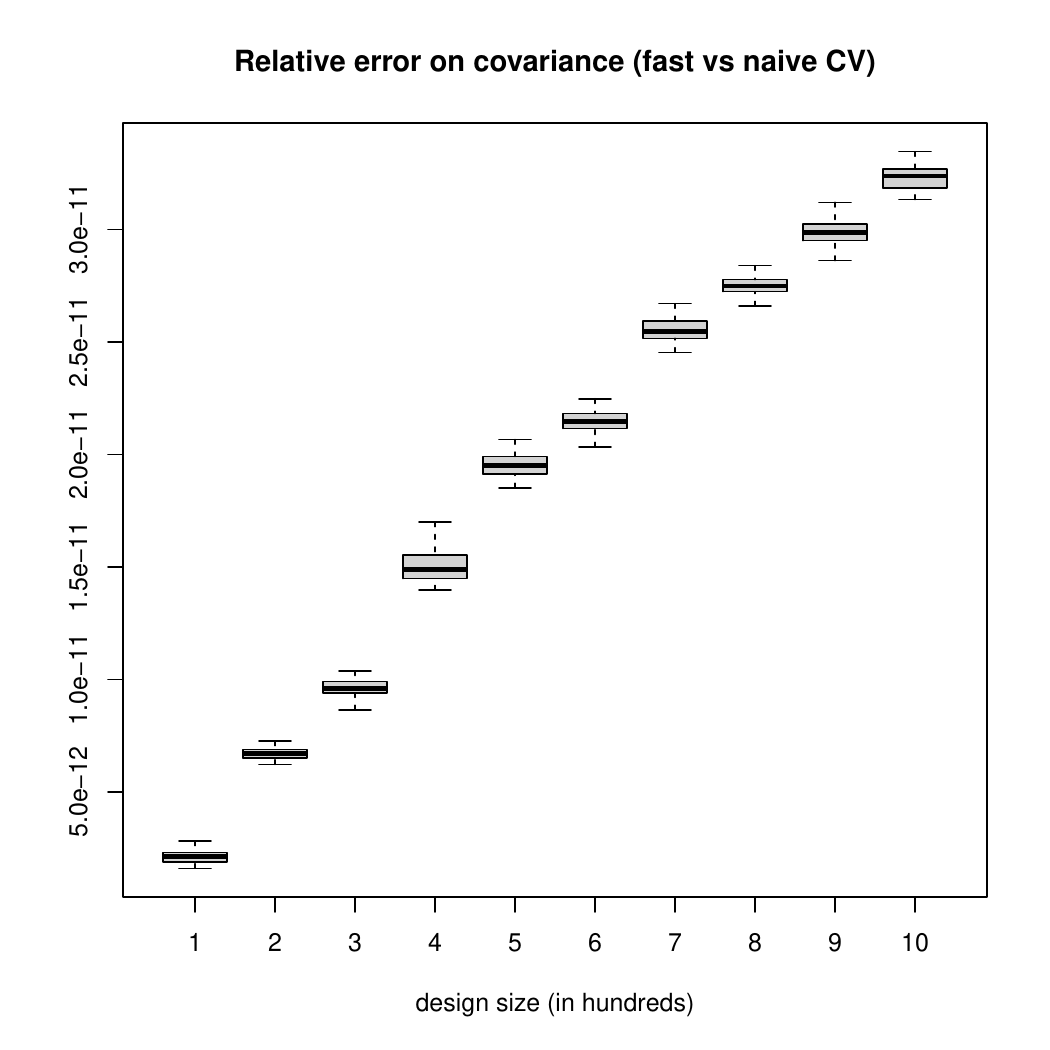}
\caption{Relative errors on CV means and covariances (between naive and fast methods, with the naive method as reference) measured as on Figure~\ref{speedupCV}.}
	\label{speedupCV_accu}
	\end{center}
\end{figure}

\subsection{More on complexity and efficiency gains}
\label{complexity}
 
We now turn to a more specific examination on how our closed form formulae affect computational speed compared to a naive implementation of multiple fold cross-validation. To this aim, we focus on the particular case of homegeneous fold sizes and denote by $r\geq 1$ their common number of elements, so that $n=r\times q$. As we will now develop, a first order reasoning on computational costs associated with naive versus closed multiple-fold cross-validation highlights the existence of two regimes depending on $q$: for large $q$ it is faster to use close form formula yet for small $q$ one should better employ the naive approach. The rationale is as follows; in a naive implementation, the main computational cost is on inverting $q$ covariance sub-matrices, all belonging to $\R^{(n-q)\times (n-q)}$. Denoting this (approximate) cost by $\kappa_{\text{naive}}$ and assuming a cubic cost for matrix inversion (with a multiplicative constant $\gamma >0$), we hence obtain that 
\begin{equation}
\kappa_{\text{naive}}=q\times \gamma (n-r)^3.
\end{equation}
One the other hand, the close form approach requires one inversion of a $n\times n$ matrix and $q$ inversions of $r \times r$ matrices, leading to an approximate cost 
\begin{equation}
\kappa_{\text{close}}=\gamma (n^3+qr^3).
\end{equation}
We thus obtain (at first order) that a speed up takes place whenever 
\begin{equation}
\begin{split}
& 0 \leq q\times \gamma (n-r)^3-\gamma (n^3+qr^3)\\
\Leftrightarrow \ & q^3\times (n-r)^3 \geq n^3 (q^2+1) \\
\Leftrightarrow \ & (q-1)^3 \geq (q^2+1), \\
\end{split}
\end{equation}
which, as $q$ is a positive integer, can be proved to be equivalent to 
$q\geq 4$, 
as $q \to (q-1)^3 - (q^2+1)$ possesses a unique real root between $3$ and $4$, takes negative values for $q\in \{1,2,3\}$, and tends to $+\infty$ when $q\to +\infty$.  
As noted in the previous section, however, in our implementation the cost of the fast approach turns out to be smaller than $\gamma (n^3+qr^3)$ since the Cholesky factor is already pre-calculated. Modelling this with a damping factor of $\alpha \in (0,1)$ in front of the $\gamma n^3$ term, we end up with an approximate cost of $\kappa_{\text{close}}=\gamma (\alpha n^3+qr^3)$, leading via analogue calculations to a speed-up whenever $(q-1)^3 \geq (\alpha q^2 +1)$. Again, the polynomial involved possesses a single real-valued root, with a value shrinking towards $2$ as $\alpha$ tends to $0$. 

\bigskip 

In practice, speed-ups are already observed here from $q=2$, which may be attributed to various reasons pertaining notably to implementation specifics and arbitrary settings used in the numerical experiments. Also, let us stress again that the reasoning done above is to be taken as an approximation at ``first order'' in the sense that it does not account for matrix-vector multiplication costs, storage and communication overheads, and further auxiliary operations that could influence the actual running times. We now present indeed some numerical experiments at fixed design size but varying number of folds and monitor the variations of the measured speed-ups and accuracy in calculating cross-validation residual means and covariances.   

\bigskip 

For convenience, we consider here a design of size $2^{10}=1024$, and let the number of folds decrease from $q=1024$ to $q=2$ by successive halvings. This amounts to folds of sizes $r$ ranging from $1$ to $512$, correspondingly. For each value of $q$, $50$ random replications of the folds are considered; this is done by permuting the indices at random prior to arranging them in regularly constructed folds of successive $r$ elements.
The resulting speed-up distributions are represented in Figure~\ref{speedup_pow2}, in function of $q$. With speed-up medians ranging from $1027.82$ to $1.50$ (means from $1429.63$ to $1.28$, all truncated after the second digit), we hence observed a speed-up whatever the value of $q$, yet with a substantial decrease when $q$ decreases compared to the LOO situation, as could be expected.  

\begin{figure}[h!]
	\begin{center}
		\includegraphics[width=.65\linewidth]{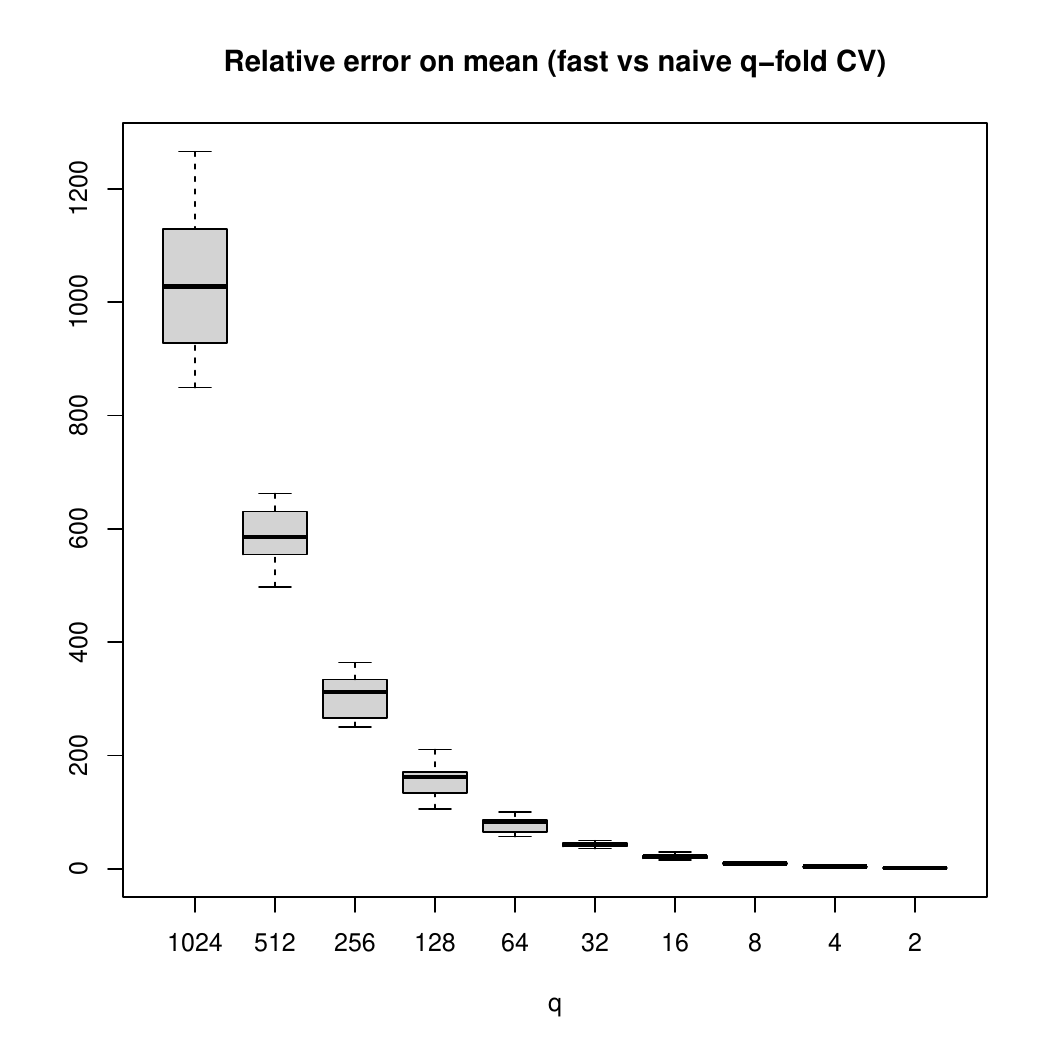}	
\caption{Speed-up (ratio between times required to run the naive and fast methods) measured for $q$-fold CV, where $q$ decreases from $1024$ to $2$ and $50$ seeds are used that affect here both model fitting and the folds.}
	\label{speedup_pow2}
	\end{center}
\end{figure}

\bigskip

Additionally, relative errors on cross-validation residuals and associated variances calculated using the fast versus naive approach are represented in Figure~\ref{relerr_pow2}, also in function of $q$. It is remarkable that, while the relative error on the calculation of residuals increases moderately (with $r$) from a median of around $3.5\times 10^{-14}$ to values nearing $4\times 10^{-14}$, in the case of the covariances the variation is more pronounced, with an increase from approx. $2\times 10^{-11}$ to $1.2\times 10^{-10}$. Further analysis may shed light on the underlying phenomena, yet one can state that with a maximum relative error of magnitude $10^{-10}$, the observed differences remain of a neglectable extent.    

\begin{figure}[h!]
	\begin{center}
		\includegraphics[width=.45\linewidth]{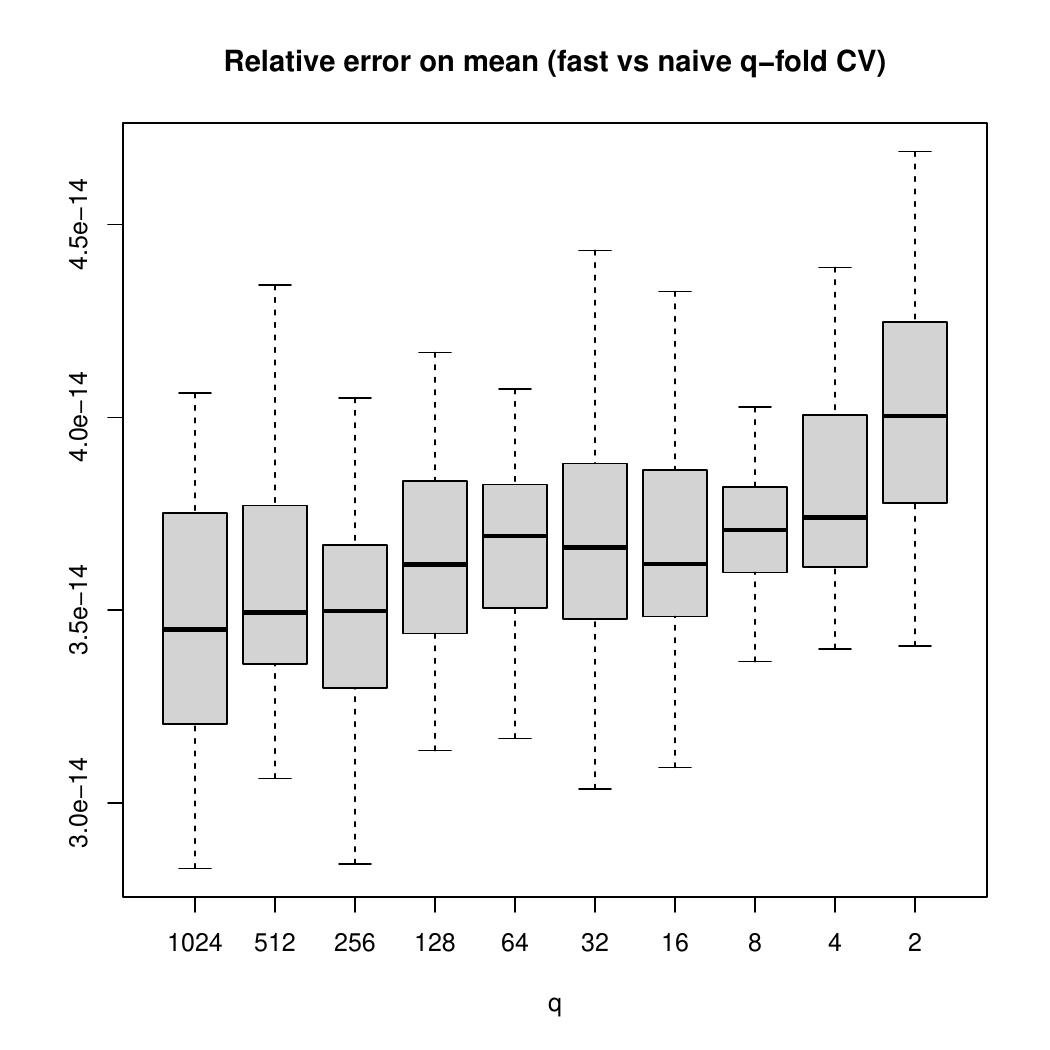}	
		\includegraphics[width=.45\linewidth]{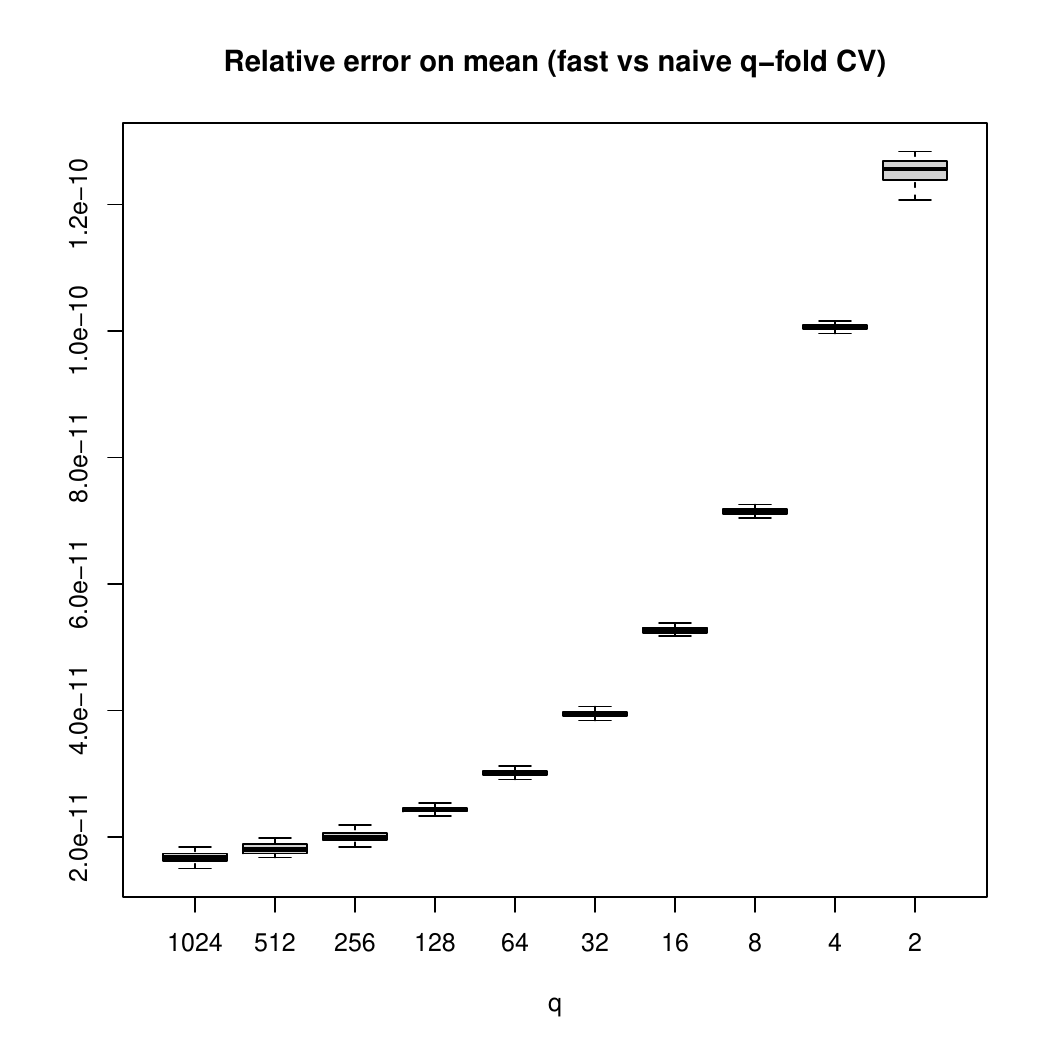}
\caption{Relative errors on CV means and covariances (between naive and fast methods, with the naive method as reference) measured in similar settings as on Figure~\ref{speedup_pow2}.}
	\label{relerr_pow2}	
	\end{center}
\end{figure}

Let us point out that the computational benefits highlighted in the case of multiple-fold cross-validation with a unique partitioning do carry over and even further grow in the case of $N\geq 2$ replications with varying partitions. Sticking to the previously considered settings but replicating the partitions, we find that a naive approach has a cost driven by 
\begin{equation}
\kappa_{\text{naive}}^{\text{rep}}=Nq\gamma (n-r)^3,
\end{equation}
while a similar approach using the fast formulae would have a cost driven by 
\begin{equation}
\kappa_{\text{close}}^{\text{rep}}=\gamma (n^3+Nqr^3).
\end{equation}
Then we obtain that $\kappa_{\text{close}}^{\text{rep}} \leq \kappa_{\text{naive}}^{\text{rep}}$ as soon as $(q-1)^3 \geq \frac{q^2}{N}+1$, occuring already from $q=3$. 

\section{Application test case}
\label{conta}

%
We now illustrate and investigate some benefits of multiple-fold cross-validation on an application test case motivated by a contaminant source localization problem from hydrogeology \cite{Pirot.etal2019}. There the goal is to minimize an objective function that is quantifying the discrepancy between given contaminant concentrations at monitoring wells and analogue quantities obtained under varying scenarios regarding the potential source underlying this contamination. Solving this difficult optimization problem in order to localize the source is performed in \cite{Pirot.etal2019} via Bayesian Optimization, and Gaussian Process models are thus used to surrogate the objective function. An instance of such an objective function is presented in Figure~\ref{conta1}. Our focus is on fitting a GP model to this function under a clustered experimental design, and to highlight notable differences arising here between leave-one-out and multiple-fold cross validation using the cluster structure. 

\begin{figure}[h!]
	\begin{center}
		\includegraphics[width=.65\linewidth]{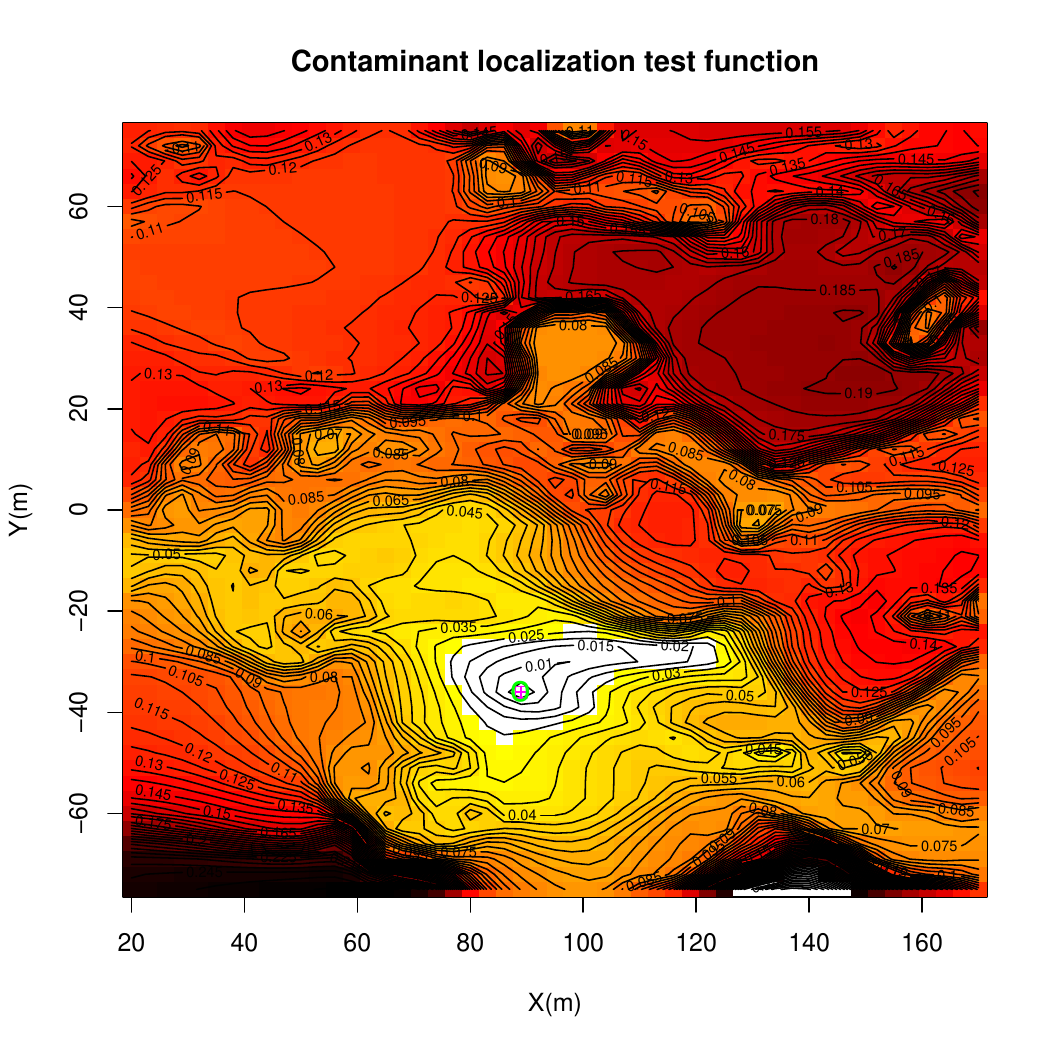}
		\caption{Contaminant localization test function designed by summing misfits between given concentrations at monitoring wells and corresponding simulation results when varying the candidate source localization (See \cite{Pirot.etal2019} for more detail).}
	\label{conta1}
	\end{center}
\end{figure}

We present in Figure~\ref{conta2-3} the mean and standard deviation maps of a GP model fitted to the latter objective function based on an experimental design consisting of $25$ $5$-element clusters. 
On the right panel of the same figure one sees how the prediction uncertainty increases when moving away from these clusters. Here the covariance parameters are estimated by MLE. This model will serve as starting point to the following discussion on LOO and multiple-fold CV.  

\begin{figure}[h]
	\begin{center}
		\includegraphics[width=.45\linewidth]{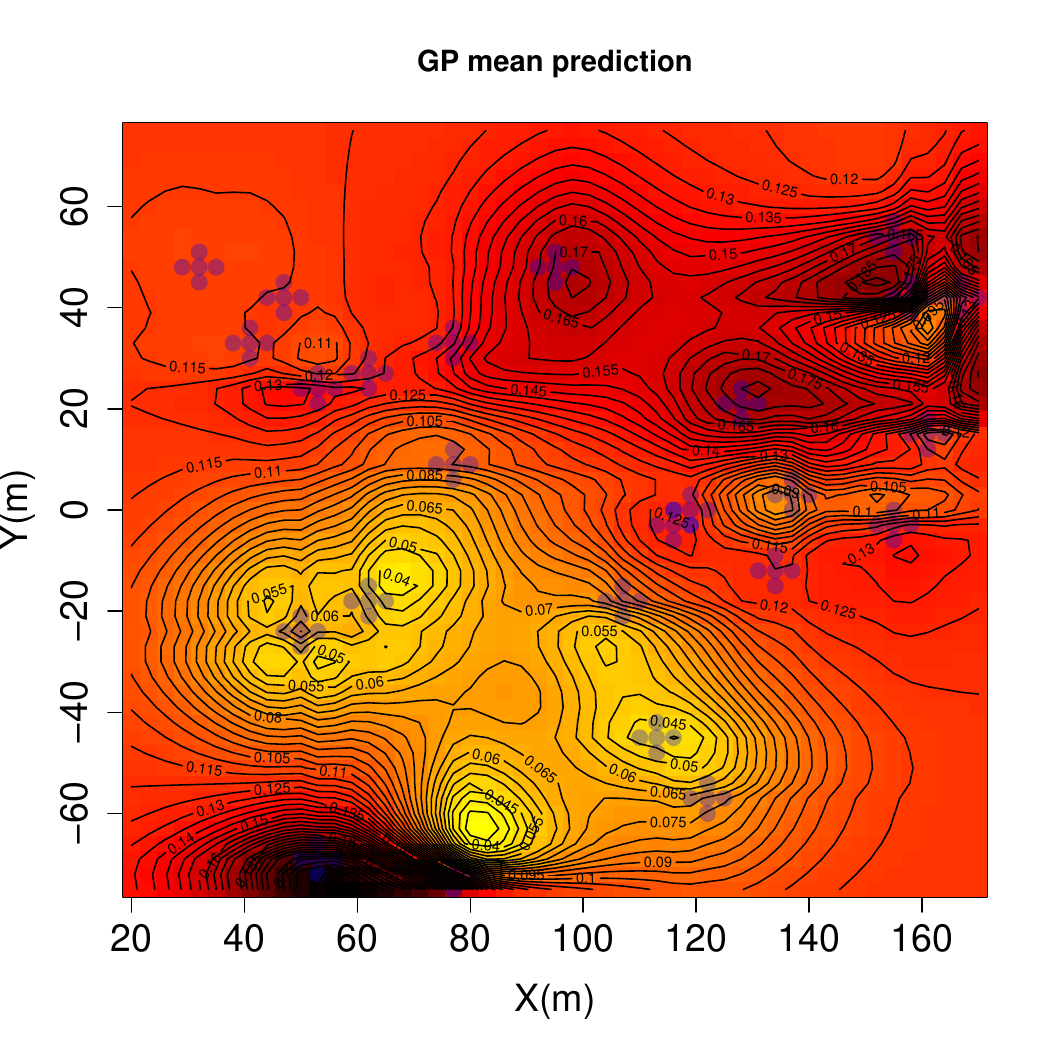}	
		\includegraphics[width=.45\linewidth]{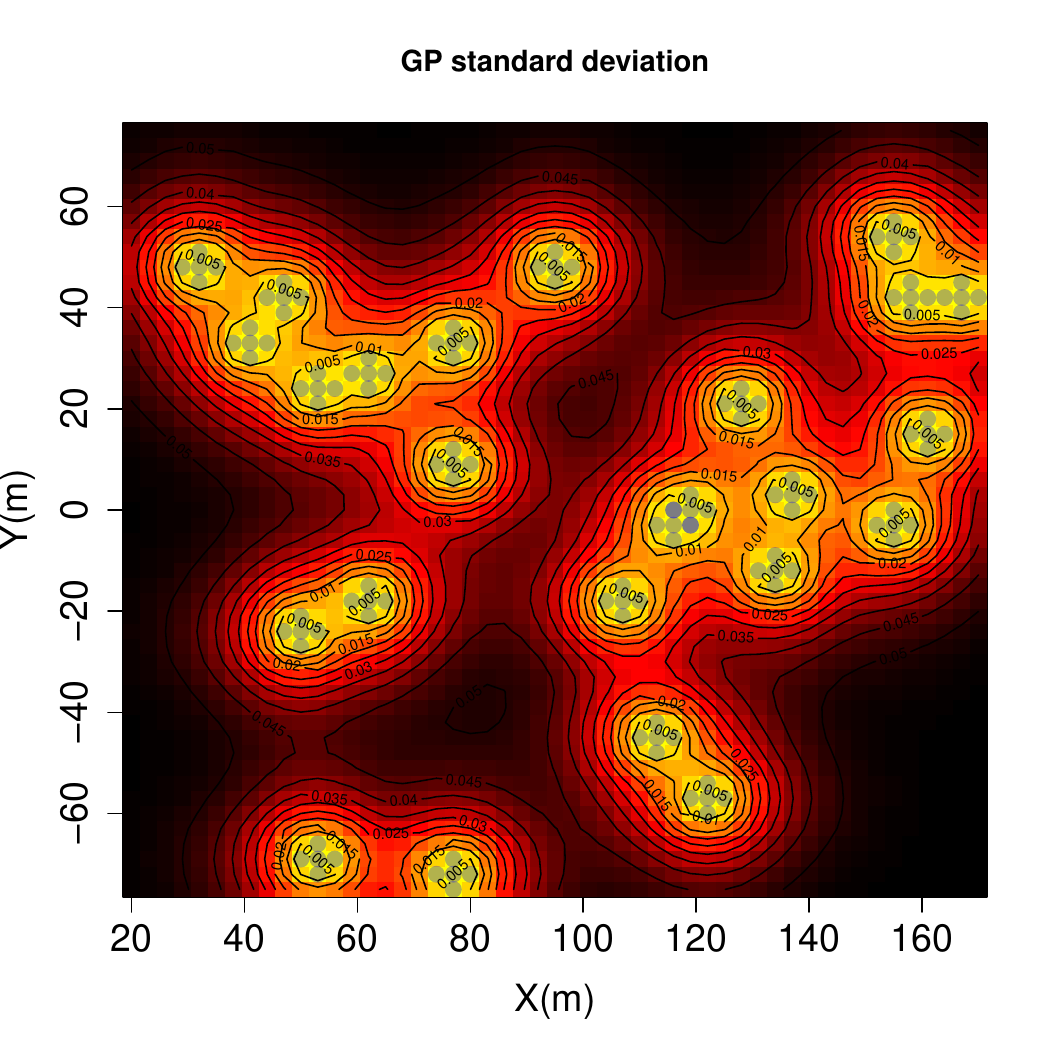}		
		\caption{Gaussian Process prediction mean and standard deviation on the contaminant localization test function with $25$ clover-shape $5$-element observation clusters and covariance parameters estimated by MLE.}
	\label{conta2-3}
	\end{center}
\end{figure}

We now compare CV in the LOO versus in ``Leave-One-Cluster-Out'' settings (i.e. $25$-fold CV with folds$\equiv$clusters). Since we can evaluate the actual function on a fine grid, we can easily calculate the absolute prediction errors and compare them to absolute cross-validation residuals. 

\begin{figure}[h!]
	\begin{center}
		\includegraphics[width=.65\linewidth]{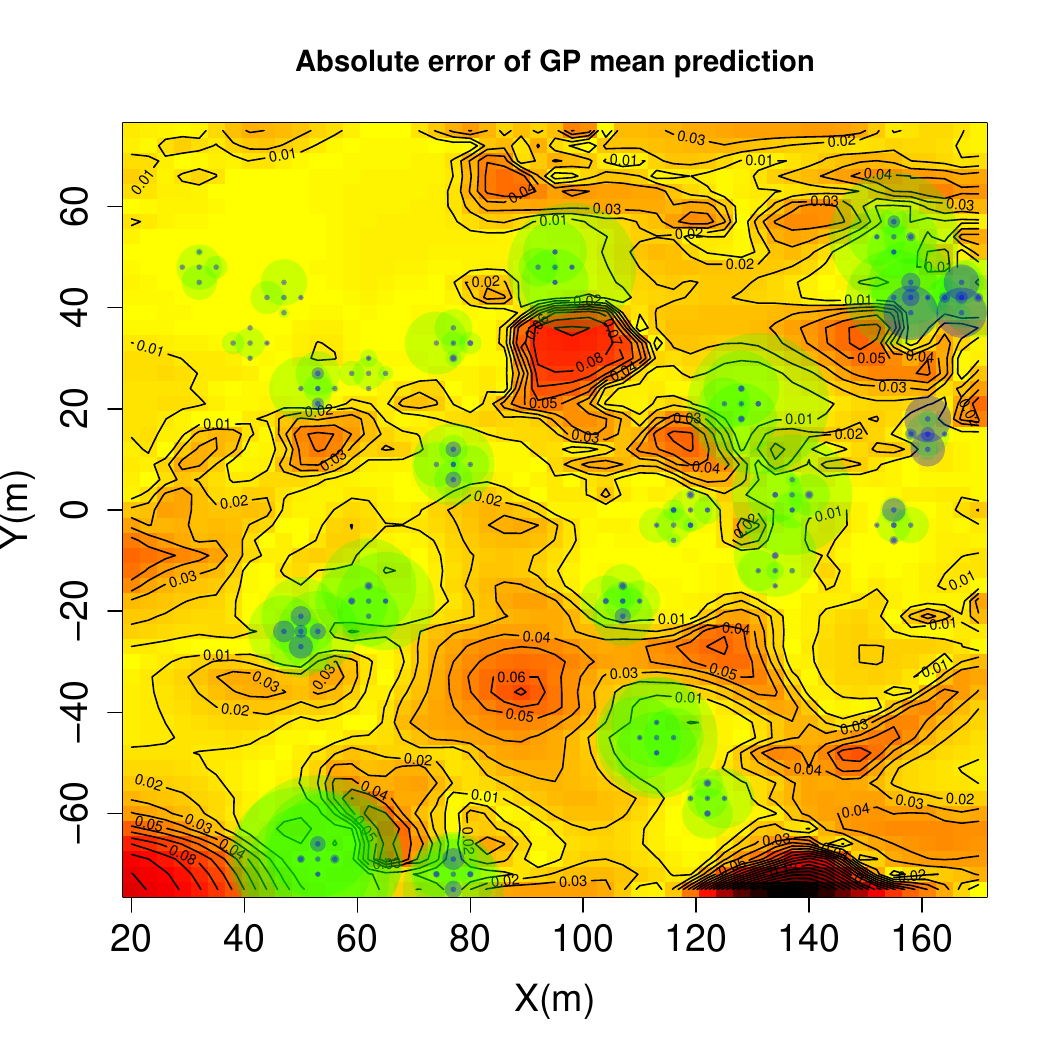}	
		\caption{Absolute prediction errors (heatmap) versus CV residuals (disks of radii proportional to absolute residuals, blue for LOO and green for MFCV).} 
		\label{conta5-6}
	\end{center}
\end{figure}

On Figure~\ref{conta5-6}, we represent a heatmap of the absolute prediction errors, to which are superimposed disks respectively representing the absolute LOO residuals (in blue) and absolute multipl-fold CV residuals (in green). It is noticeable that the LOO absolute residuals are in this case rarely informative about the presence of higher prediction errors nearby while the multiple-fold CV absolute residuals tend to better point out regions where the model is comparatively poorly performing.  

Figure~\ref{conta5-6bis} illustrate by representing LOO versus MFCV residual norms as a function of the range parameter and comparing the to the true reconstruction error that, here again, properly designing the folds drastically change the outcome. In the present case, it can be observed that the curve of MFCV absolute residuals much more closely follows the actual reconstruction error, pointing out in particular a better suited value of the range parameter while displaying a higher curvature. In contrast, the top curve associated with LOO features a flatter region near the optimum resulting in turn in an optimal range associated with a larger reconstruction error. 

\begin{figure}[h!]
	\begin{center}
		\includegraphics[width=.65\linewidth]{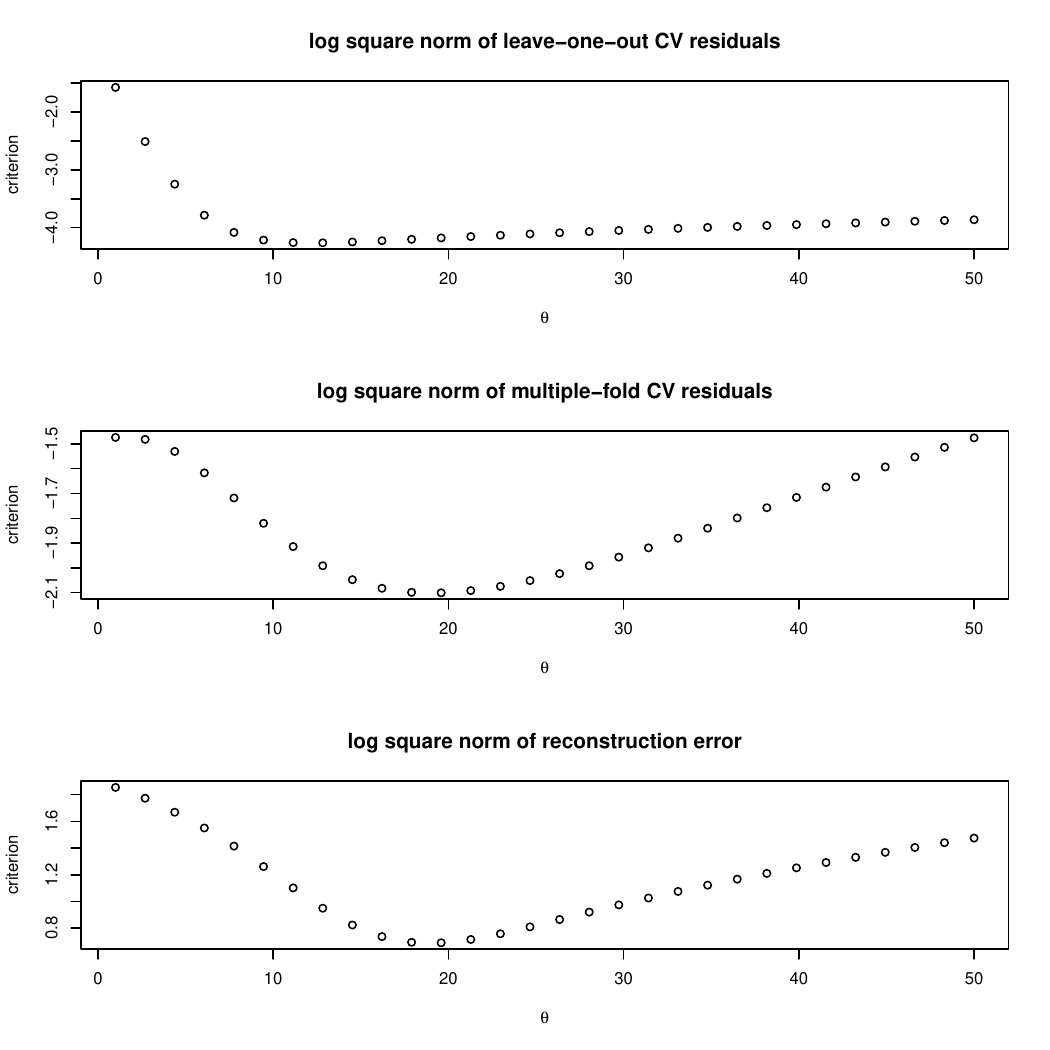}		
		\caption{
		Right: log square norm of LOO residuals (top), of MFCV residuals (center), and of reconstruction error (bottom) as a function of the range hyperparameter.}
		\label{conta5-6bis}
	\end{center}
\end{figure}

\section{Discussion}
\label{discussion}

Fast leave-one-out formulae for Gaussian Process prediction that have been used for model diagnostics and hyperparameter fitting do carry over well to multiple-fold cross validation, as our theoretical and numerical results in both Simple and Universal Kriging settings suggest. The resulting formulae were found to allow substantial speed-ups, yet with a most favourable scenario in the leave-one-out case and decreasing advantages in cases with lesser folds of larger observation subsets. A first order analysis in terms of involved matrix inversion complexities confirmed the observed trend, yet in most considered situations fast formulae appeared computationally worthwhile. In addition, established formulae enabled a closed-form calculation of the covariance structure of cross-validation residuals, and eventually correcting for an improper assumption of independence silently underlying QQ-plots with standardized LOO residuals.   

\bigskip 

As we established as well, the established formulae lend themselve well to generalizing existing covariance hyperparameter estimation approaches that are based on cross-validation residuals. Looking first at scale parameter estimation in the LOO case, we found that changing a criterion underlying the scale estimation approach used in \cite{Santner.etal2003, Bachoc2013} by correcting for covariance between LOO residuals led back to the maximum likelihood estimator of scale. On a different note and with more general covariance hyperparameters in view, as established in Corollary~\ref{cor2} maximizing the pseudo-likelihood criterion mentioned in \cite{Rasmussen.Williams2006} was found to coincide with MLE in the case of independent cross-validation residuals.  

\bigskip

It is interesting to note as a potential starting point to future work on cross-validation-based parameter estimation that going beyond the independence assumption of Corollary~\ref{cor2} and accounting for the calculated residual covariances results in fact in a departure from basic Maximum Likelihood Estimation. Still assuming indeed that $\Delta=I_{n}$ and replacing the product in $C_{\text{CV}}^{(3)}(\psi)$ by the joint density of cross-validation residuals, we end up indeed with
\begin{equation}
\label{CV3mod}
\begin{split}
C_{\widetilde{\text{CV}}}^{(3)}(\psi)
&
=\log\left( 
p_{(\mathbf{E}_{\ii_1}, \dots, \mathbf{E}_{\ii_q})}((\mathbf{e}_{\ii_1}, \dots, \mathbf{e}_{\ii_q}); \psi) 
\right)\\
&
=\log\left( 
p_{\D \covsum^{-1} \mathbf{Z}}(\mathbf{e}; \psi) 
\right)\\
&
=\log\left( 
p_{\D \covsum^{-1} \mathbf{Z}}(\bloc \covsum^{-1}\mathbf{z}; \psi) 
\right)\\
&=\log\left( 
p_{\mathbf{Z}}(\mathbf{z}; \psi) \vert \det(\bloc \covsum^{-1})  \vert  
\right)\\
&=\log\left( 
p_{\mathbf{Z}}(\mathbf{z}; \psi) \right)
+ \log\left( \det(\bloc) \right) - \log\left( \det(\covsum) \right),
\end{split}
\end{equation} 
and we see that maximizing $C_{\widetilde{\text{CV}}}^{(3)}(\psi)$ generally departs indeed from MLE, as $\covsum$ and $\bloc$ are functions of $\psi$, yet coincides with it when $\det(\bloc)=\det(\covsum)$.
While studying further this new criterion is out of scope of the present paper, we believe that using the full distribution of cross-validation residuals could be helpful in designing novel criteria and procedures for covariance parameter estimation and more.  
Beyond this, efficient adjoint computation of gradients could be useful to speed-up covariance hyperparameter estimation procedures such as studied here, but also variations thereof relying on further scoring rules \cite{Petit}, not only in LOO settings but also potentially in the case of multiple-fold cross validation. 

\bigskip 

It is important to stress that while the presented results do not come with procedures for the design of folds, they might be of interest to create objective functions for their choice. Looking at the covariance matrix between cross-validation residuals might be a relevant ingredient to such procedures. On the other hand future research might also be concerned with choosing folds so as to deliver cross-validation residuals reflecting in some sense the generalization errors of the considered models. Of course, designing the folds is not independent of choosing the $\xx_i$'s and we expect interesting challenges to arise at this interface. Also, stochasic optimization under random folds is yet another approach of interest which could be eased thanks to fast computation of cross-validation residuals.  
 Last but not least, as fast multiple-fold cross-validation carries over well to linear regression and GP modelling shares foundations with several other paradigms, it would be desirable to generalize presented results to broader model classes (see for instance \cite{Rabinowicz2020a}, where mixed models are tackled).

\section*{Acknowledgements}

Part of DG's contributions have taken place within the Swiss National Science Foundation project number 178858, and he would like to thank Ath\'ena\"is Gautier, Fabian Guignard, C\'edric Travelletti and Riccardo Turin for stimulating scientific exchanges. Yves Deville should be warmly thanked too for insightful comments on an earlier version of this paper. Warm thanks as well to the two referees and the members of the editorial team, whose comments have enabled improving the paper substantially. Calculations were performed on UBELIX, the High Performance Computing (HPC) cluster of the University of Bern. DG would also like to aknowledge support of Idiap Research Institute, his primary affiliation during earlier stages of this work. 

\bibliographystyle{plain}

\appendix 

\section{About block inversion via Schur complements}
\label{app_schur}

The following, partly based on \cite{Gallier}, is a summary of standard results revolving around the celebrated notion of Schur complement.  

\begin{theorem}
	\label{Schur_inv}
	Let $M=\left( \begin{matrix} A & B \\ C & D\end{matrix} \right)$ be a real $n\times n$ matrix with blocks $A,B,C,D$ of respective sizes $p\times p$, $p\times q$, $q\times p$ and $q\times q$, where $p,q\geq 1$ such that $n=p+q$. Assuming that $D$ and $A-BD^{-1} C$ are invertible, then so is $M$ with 
	\begin{equation*}
	M^{-1}=\left( \begin{matrix} 
	(A-BD^{-1} C)^{-1} & -(A-BD^{-1} C)^{-1}BD^{-1}  \\ 
	-D^{-1}C (A-BD^{-1} C)^{-1} & D^{-1}+D^{-1}C(A-BD^{-1} C)^{-1}BD^{-1}
	\end{matrix} \right).
	\end{equation*}
\end{theorem}

\begin{proof}
	We follow the classical path consisting of solving for $z=\left( \begin{matrix} x & y\end{matrix} \right)^{T}$ in the equation $Mz=w$, with $w=\left( \begin{matrix} u & v\end{matrix} \right)^{T} \in \mathbb{R}^n$ where $x,u\in \mathbb{R}^p$ and $y,v\in \mathbb{R}^q$. Using the assumed invertibility of $D$, we have indeed equivalence between 
	
	$$
	\left\{
	\begin{array}{ll}
	Ax + By&=u\\
	Cx + Dy&=v
	\end{array}
	\right.
	\text{
		and 
	}
	\left\{
	\begin{array}{l}
	y=D^{-1}(v-Cx)\\
	(A-BD^{-1}C)x=u-BD^{-1}v
	\end{array}
	\right. ,
	$$
	whereof, using this time the assumed invertibility of $(A-BD^{-1} C)$, 
	$$
	\left\{
	\begin{array}{ll}
	x&=(A-BD^{-1} C)^{-1}u - (A-BD^{-1} C)^{-1} BD^{-1} v\\
	y&=- (A-BD^{-1} C)^{-1} D^{-1}C u + (D^{-1} + D^{-1}C(A-BD^{-1} C)^{-1} BD^{-1}) v
	\end{array}
	\right. ,
	$$
	resulting in the claimed result for $M^{-1}$.
\end{proof}

\begin{remark} 
	The block inversion formula of Theorem~\ref{Schur_inv} above can be reformulated as the following product of three matrices
	$$
	M^{-1}=
	\left( \begin{matrix} 
	I & 0 \\
	-D^{-1}C &I
	\end{matrix} \right)
	\left( \begin{matrix} 
	(A-BD^{-1} C)^{-1} & 0 \\ 
	0 & D^{-1}
	\end{matrix} \right)
	\left( \begin{matrix} 
	I & -BD^{-1} \\
	0 &I
	\end{matrix} \right),
	$$  
	implying in turn the following, that actually only requires the invertibility of $D$,
	$$
	M=
	\left( \begin{matrix} 
	I & BD^{-1} \\
	0 &I
	\end{matrix} \right)
	\left( \begin{matrix} 
	(A-BD^{-1} C) & 0 \\ 
	0 & D
	\end{matrix} \right)
	\left( \begin{matrix} 
	I & 0 \\
	D^{-1}C &I
	\end{matrix} \right),
	$$ 
	and illustrates that the simultaneous invertibility of $D$ and $A-BD^{-1} C$ is actually equivalent to the invertibility of $M$ itself. 
\end{remark}

\begin{remark} 
	\label{rkpermut}
	As classically, in Theorem~\ref{Schur_inv} the inverted block is the bottom-right $D$ and correspondingly the considered Schur complement is the one of the upper left block $A$. The Schur inverted Schur complement $(A-BD^{-1} C)^{-1}$ hence appears as the upper left block of $M^{-1}$ associated with $\ii=(1,\dots,p)$. Yet, analogue derivations are possible for any vector of indices $\ii$ from $\{1,\dots, n\}$. Reformulating a result presented in Horn and Johnson in terms of our notation, we have indeed  
	\begin{equation}
	\label{horn1}
	M^{-1}[\ii]=(M[\ii]-M[\ii,-\ii]M[-\ii]^{-1}M[-\ii,\ii])^{-1}
	\end{equation}
	and, more generally for two vectors of indices $\ii$ and $\jj$, 
	\begin{equation}
	\begin{split}
	\label{horn2}
	M^{-1}[\ii,\jj] & =-(M[\ii]-M[\ii,\jj]M[-\ii]^{-1}M[\jj,\ii])^{-1}M[\ii,\jj]M[\jj]^{-1}
	\\
	&= -M[\jj]^{-1}M[\jj,\ii](M[\ii]-M[\ii,\jj]M[\jj]^{-1}M[\jj,\ii])^{-1}.
	\end{split}
	\end{equation}
\end{remark}

\begin{remark} 
	\label{schur_variant}
	Following up on Remark~\ref{rkpermut} and using the notation of 
	Theorem~\ref{Schur_inv}, we obtain in particular when assuming that $A$ and $(D-CA^{-1}D)$ are invertible and considering indeed the Schur complement of $A$ instead of $D$'s one that 
	\begin{equation}
	\label{schur_variant_eq}
	M^{-1}=\left( \begin{matrix} 
	A^{-1}+A^{-1}B(D-CA^{-1} B)^{-1}CA^{-1} & -A^{-1}B (D-CA^{-1} B)^{-1} \\ 
	-(D-CA^{-1} B)^{-1}CA^{-1} & (D-CA^{-1} CB)^{-1}
	\end{matrix} \right).
	\end{equation}
	Assuming that both blocks $A,D$ and their respective Schur complements are invertible, one then retrieves the so-called binomial inverse theorem, namely 
	$$
	(D-CA^{-1}D)^{-1}=D^{-1}+D^{-1}C(A-BD^{-1} C)^{-1}BD^{-1}.
	$$ 	
	As a by-product, equating non-diagonal blocks further delivers that 
	$$
	(A-BD^{-1} C)^{-1}BD^{-1}
	=
	A^{-1}B (D-CA^{-1} B)^{-1}
	$$
	and similarly, 
	$$
	D^{-1}C (A-BD^{-1} C)^{-1}
	=
	(D-CA^{-1} B)^{-1}CA^{-1}
	.
	$$
\end{remark}


\section{More on the $1$-dimensional analytical example} 
\label{app:example}

As a first example, we consider a case of GP prediction for a one-dimensional test function from \cite{Xiong.etal2007}, 
\begin{equation}
\label{xiong}
f: (x)\in [0,1] \mapsto 
f(x)=\mathrm{sin}(30(x-0.9)^4)\mathrm{cos}(2(x-0.9)) + (x-0.9)/2.
\end{equation}
For simplicity we consider a regular design, here a 10-point subdivision of $[0,1]$, and a Simple Kriging predictor assuming a Mat\'ern 5/2 stationary covariance kernel. The test function is represented in black in Figure~\ref{xiong_preds}, along with the GP predictor (blue line) and Leave-One-Out (LOO) predictions at the design locations (red points). Modelling is performed with the \textrm{DiceKriging} package \cite{Roustant.etal2012} and cross-validation relies on a fast implementation using the presented results and the \textrm{cv} function newly available in the \textrm{DiceKriging} package (version 1.6.0). The accuracy and speed-ups associated with this fast implementation are analysed in more detail in Section~\ref{speedup}. 
Before presenting further consequences of Theorem~\ref{thm1} on this example, let us examine how LOO cross-validation compares on this example with actual prediction errors and predictive standard deviations stemming from the GP model. 

\bigskip

By design, the considered example function presents moderate fluctuations in the first third of the domain, then a big up-and-down in the second third, and finally is relatively flat in the last third. 
Of course, appealing to a non-stationary GP that could capture this is legitimate in such a case (and has been done, for instance in \cite{Xiong.etal2007, Marmin2018}), yet we rather stick here to a baseline stationary model and start by illustrating how cross-validation highlights this heterogeneity. Looking first at the upper panel of Figure~\ref{xiong_preds}, we see that LOO residuals tend to take their larger magnitudes on the boundaries and around the center of the domain. A closer look reveals that the residual at the left boundary point has a higher magnitude than the one of the right boundary, and that the location with the highest LOO residual is at the top of the central bump. The graph on the lower panel of Figure~\ref{xiong_preds} represents the absolute prediction errors associated with the baseline (``basic'') GP model (black line), with the LOO cross-validation (red points), as well as the prediction standard deviation coming with the basic GP model (blue line). While the GP prediction standard deviation appears to behave in an homogeneous way across the domain (it is known to depend solely on the evaluation locations and not on the responses, a phenomenon referred to as \textit{homoscedasticity in the observations}), the true absolute prediction errors are concentrated on left half of the domain, with a largest peak around the left boundary and several bumps spreading out towards the center of the domain. Contrarily to what the LOO magnitude at the end right point suggests, there is no substantial prediction error in this region of the domain. In all, while LOO versus actual absolute prediction errors are not in complete agreement, LOO signals more difficulty to predict in the left domain half, something that the prediction standard deviation attached to the considered stationary GP model is by design unable to detect. 

\subsection{Extending the example with multiple-fold cross-validation}
\label{ex2_fastmfcv}

We now mimick the previous example but consider instead of our $10$ previous points a design of $20$ points formed by $10$ pairs of neighbouring points. First, a regular design of $10$ points between $\delta$ and $1-\delta$ is formed, where $\delta$ is a constant with value prescribed to $0.001$ in this example. Then, for every point of this base design, two design points are created by adding realizations of independently generated uniform random variables on $[-\delta,\delta]$. 

\begin{figure}[h!]
	\begin{center}
		\includegraphics[width=.65\linewidth]{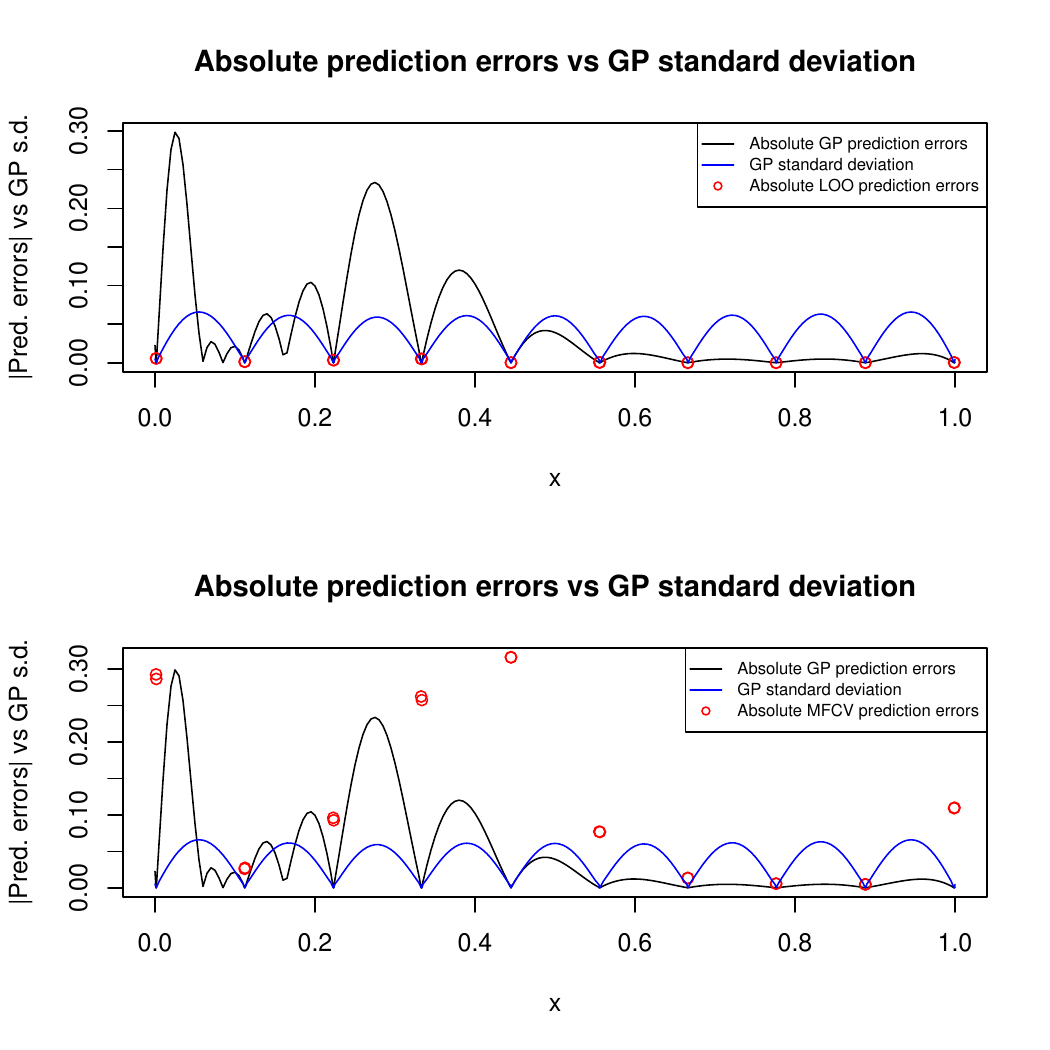}
		\caption{Absolute prediction errors associated with GP (black line) and cross-validation (red point) predictions, and GP prediction standard deviation (in blue). Upper panel: LOO. Lower panel: 10-fold CV where each fold contains two successive points from the example, at distances upper bounded by $2\delta=2.10^{-3}$.}
		\label{xiong_LOOvsMFCV}
	\end{center}
\end{figure}

As can be seen on the upper panel of Figure~\ref{xiong_LOOvsMFCV}, LOO is here completely myopic to prediction errors, precisely because of the paired points. Indeed, when leaving one point out at a time, there is always a very close neighbour that enables making a an accurate Kriging prediction. This somehow extreme situation highlights how LOO (and more generally cross-validation with poorly designed partition of input points) may lead to a bad picture of the model's generalization ability. When perfoming multiple-fold cross-validation with grouping of the pairs, however, the residuals are much more insightful, as illustrated on the lower panel of Figure~\ref{xiong_LOOvsMFCV}. Let us further remark that Theorem~\ref{thm1} delivers in turn the respective covariance matrices between residuals within groups. 

\subsection{Correlation of LOO residuals with an additional example}
 
As an additional cross-validation example with a focus on the effect of correlation on Q-Q plots, we added to the data used in the first example (Evaluations at $10$-point regular subdivision of [0,1] of the one-dimensional test function recalled in Eq. \ref{xiong}) a cluster of $10$ additional function evaluations for equi-spaced $x$ values between $0.1$ and $0.3$, and we now present the effect of this cluster on similar Q-Q plots.
A brief comparison between Q-Q plots obtained on the second example (the one with $10$ paired points, not further developed in this section by prioritization) is appended in Section~\ref{QQCV_ex2} for completeness.

\bigskip 

We can see on Figure~\ref{xiong_predsclust} that the function is now very well captured in the region of the added cluster. Yet what is not appearent here but is reflected by Figure~\ref{xiong_QQplotsLOOclust} is that leave-one-out predictions are coming with over-evaluated prediction uncertainty for a number of locations, that turn out to originate from the region of the added cluster and also from the flatter region on the right handside. While looking at the Q-Q plot of standardized residuals on the right panel of Figure~\ref{xiong_QQplotsLOOclust} suggests some model inadequacy, accounting for and removing covariance between leave-one-out prediction errors as illustrated in the left panel of Figure~\ref{xiong_QQplotsLOOclust} suggests on the contrary that the assumptions underlying the considered GP model do not stand out as unreasonable.

\begin{figure}[h!]
	\begin{center}
		\includegraphics[width=.8\linewidth]{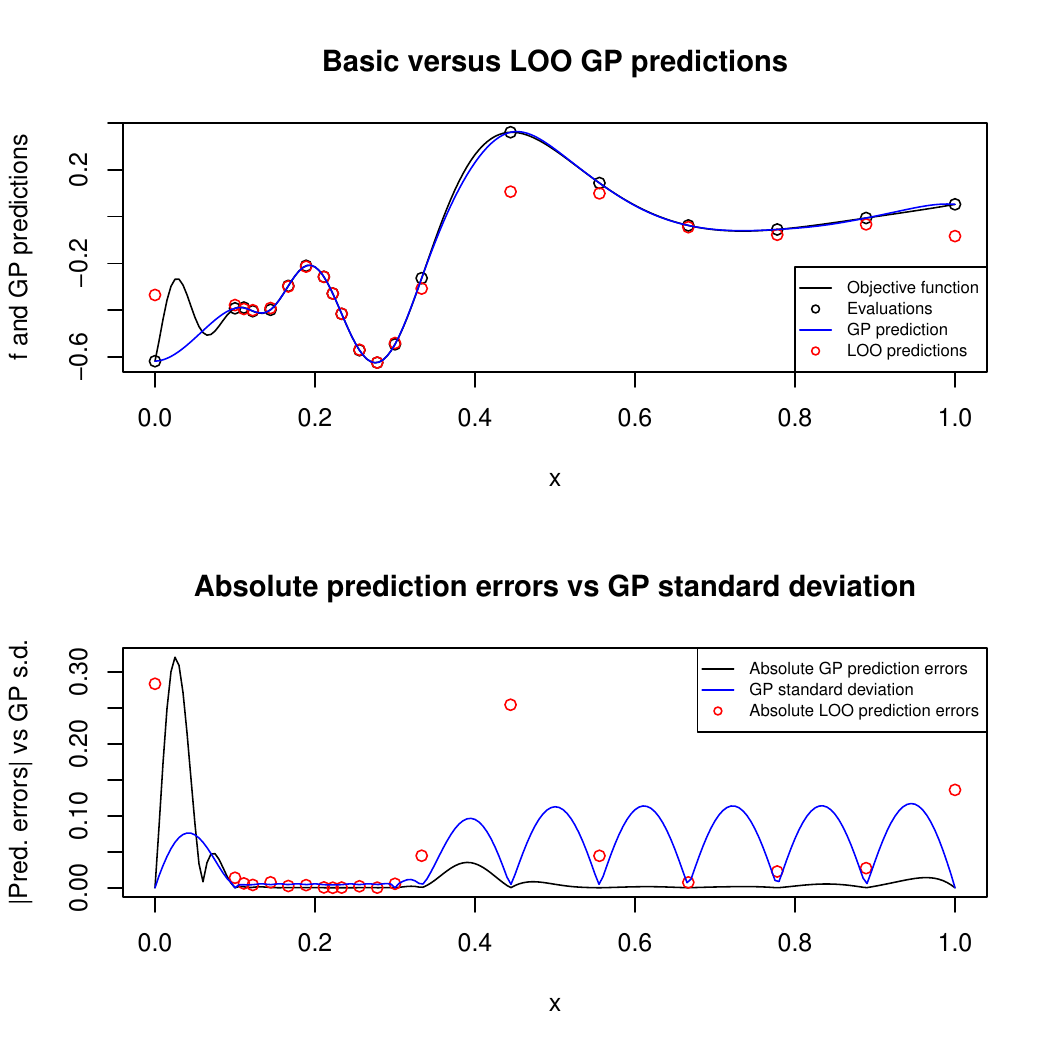}
		\caption{Upper panel: GP prediction (blue line) of the test function (black line) defined by Equation~\ref{xiong} based on evaluations at a regular grid complemented by clustered points between $0.1$ and $0.3$, LOO predictions (red points). Lower panel: absolute prediction errors associated with GP (black line) and LOO (red point) predictions, and GP prediction standard deviation (in blue).}
		\label{xiong_predsclust}
	\end{center}
\end{figure}

\begin{figure}[h!]
	\includegraphics[width=.5\linewidth]{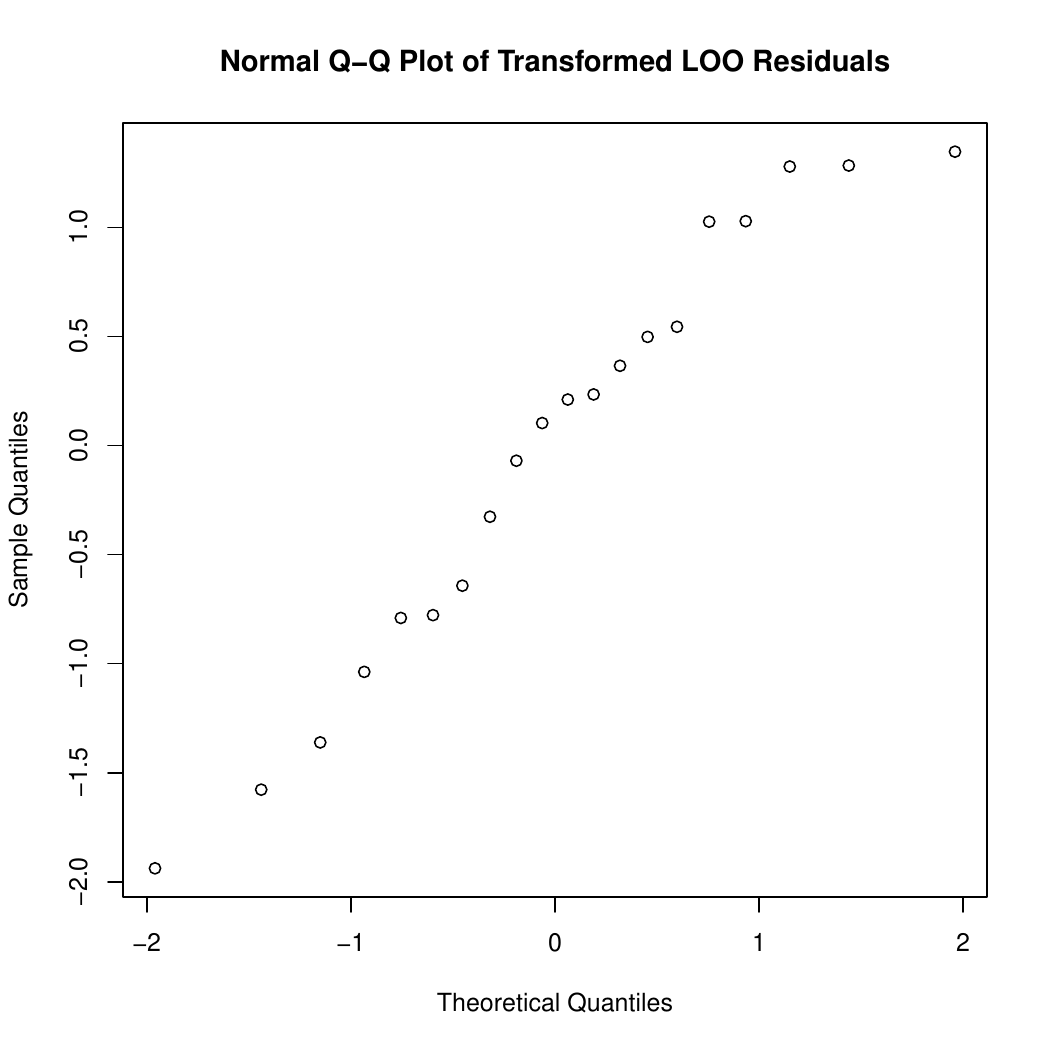}
	\includegraphics[width=.5\linewidth]{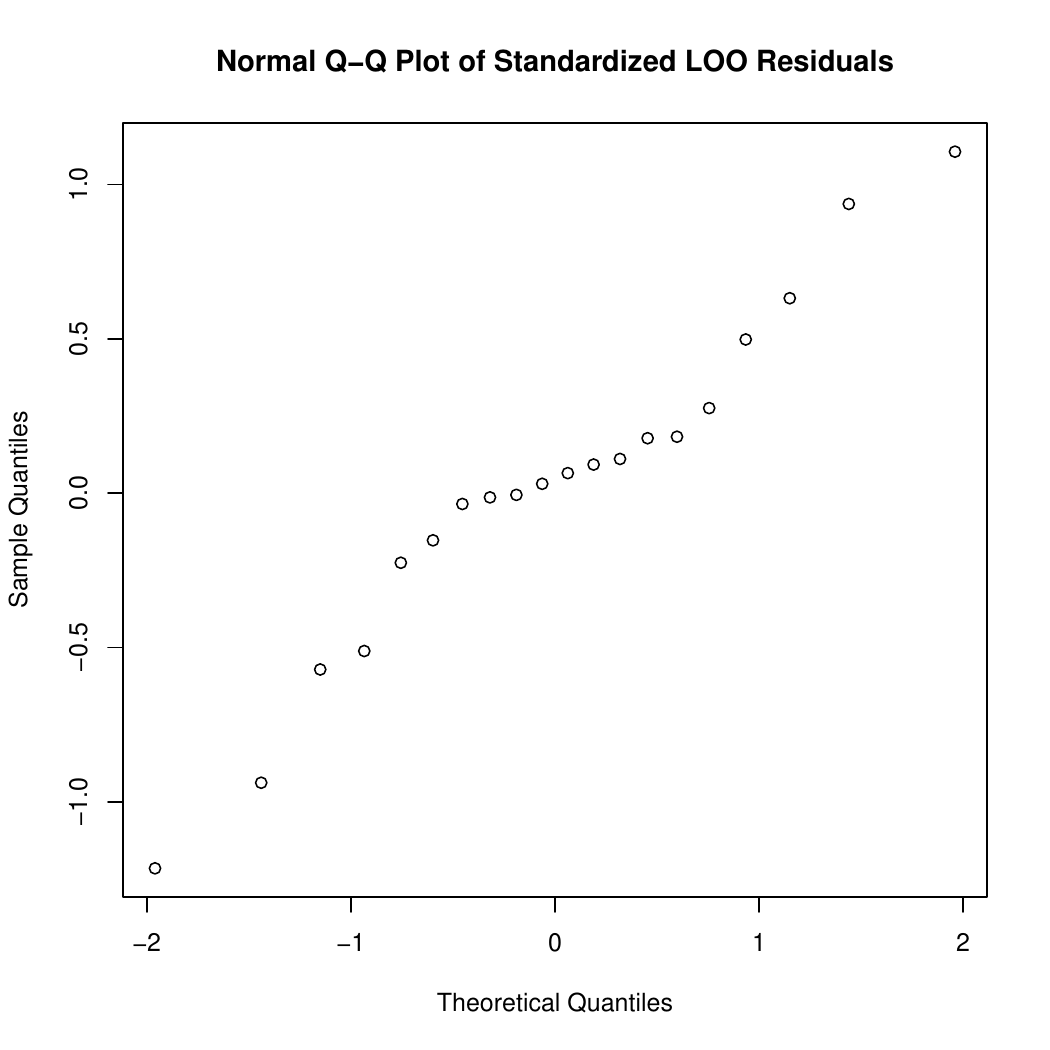}	
	\caption{On the effect of accounting and correcting for correlation in QQ-plots based on LOO residuals (second case: regular grid complemented by clustered points between $0.1$ and $0.3$). Right panel: QQ-plot against $\mathcal{N}(0,1)$ of LOO residuals merely divided by corresponding LOO standard deviations. Left panel: QQ-plot against $\mathcal{N}(0,1)$ of duly transformed LOO residuals.}
	\label{xiong_QQplotsLOOclust}
\end{figure}

\bigskip 

The results obtained on this modified test case further stress the deleterious effect of neglecting the correlation between cross-validation residuals when diagnosing GP models relying notably on Q-Q plots. 
For completeness, we also added in Section~\ref{CV_ex3} a brief analysis of what cross-validation would result in on this third example when grouping the second to eleventh points within an index vector (and keeping the other points as individual entities like in LOO).  

\newpage
\section{Additional detail on the basics of Universal Kriging}
\label{basicsUKold}

Considering further the settings of Section~\ref{basicsUK} and assuming as before that $\covsum$ and $F^{\top}\covsum^{-1}F$ are invertible and applying to $M=\left( \begin{matrix} 
\covsum & F \\
F^{\top} & 0 
\end{matrix} \right)$ the variant of the Schur complement inversion formula discussed in Remark~\ref{schur_variant}, we get that $M$ is invertible and has inverse  
\begin{equation*}
M^{-1}=
\left( \begin{matrix} 
\covsum^{-1} - \covsum^{-1} F (F^{\top}\covsum^{-1}F)^{-1} F^{\top} \covsum^{-1} & \covsum^{-1} F (F^{\top}\covsum^{-1}F)^{-1} \\
(F^{\top}\covsum^{-1}F)^{-1} F^{\top} \covsum^{-1} & - (F^{\top}\covsum^{-1}F)^{-1}. 
\end{matrix} \right) 
\end{equation*} 
Recalling that the Universal Kriging predictor at the considered $\xx$ is given by $\widehat{Z}_{\xx}=\boldsymbol{\lambda}'\Z=\Z'\boldsymbol{\lambda}$, we hence find that $\widehat{Z}_{\xx}=\boldsymbol{\alpha}'\mathbf{k}(\xx) + \widehat{\mybeta}'\mathbf{f}(\xx)$ where 
\begin{align}
\boldsymbol{\alpha}&=(\covsum^{-1} - \covsum^{-1} F (F^{\top}\covsum^{-1}F)^{-1} F^{\top} \covsum^{-1}) \Z, \\
\widehat{\mybeta}&=(F^{\top}\covsum^{-1}F)^{-1} F^{\top} \covsum^{-1}\Z.
\end{align}
As it turns out, $\widehat{\mybeta}$ coincides with the Generalized Least Squares estimator of $\mybeta$ under assumption of a (centred) noise term with covariance matrix $\covsum$. Besides this, these values for $\boldsymbol{\alpha}$ and $\widehat{\mybeta}$ are characterized by 
\begin{equation}
\label{UK_system}
\left( \begin{matrix} 
\covsum & F \\
F^{\top} & 0 
\end{matrix} \right)
\left( \begin{matrix} 
\boldsymbol{\alpha} \\
\widehat{\mybeta}
\end{matrix} \right)
=
\left( \begin{matrix} 
\Z \\
0
\end{matrix} \right),
\end{equation}
as is well-known in the geostatistics literature (See \cite{Dubrule1983} and references therein). Note that \cite{Dubrule1983} is a seminal reference that tackled fast cross-validation formulae already in the Universal Kriging case and also tackled the case of two left out points. While our formulas coincide with those presented in \cite{Dubrule1983}, the bloc formalism sheds light on the algebraic mechanisms at work and somehow help achieving a broader level of generality, as illustrated in Section~\ref{Extensions} with the relative simplicity of obtaining multiple-fold cross-validation conditional covariance matrices in the Universal Kriging framework.  


\section{LOO speed-ups}
\label{LOO_speedups}

\begin{figure}[h!]
	\begin{center}
		\includegraphics[width=.65\linewidth]{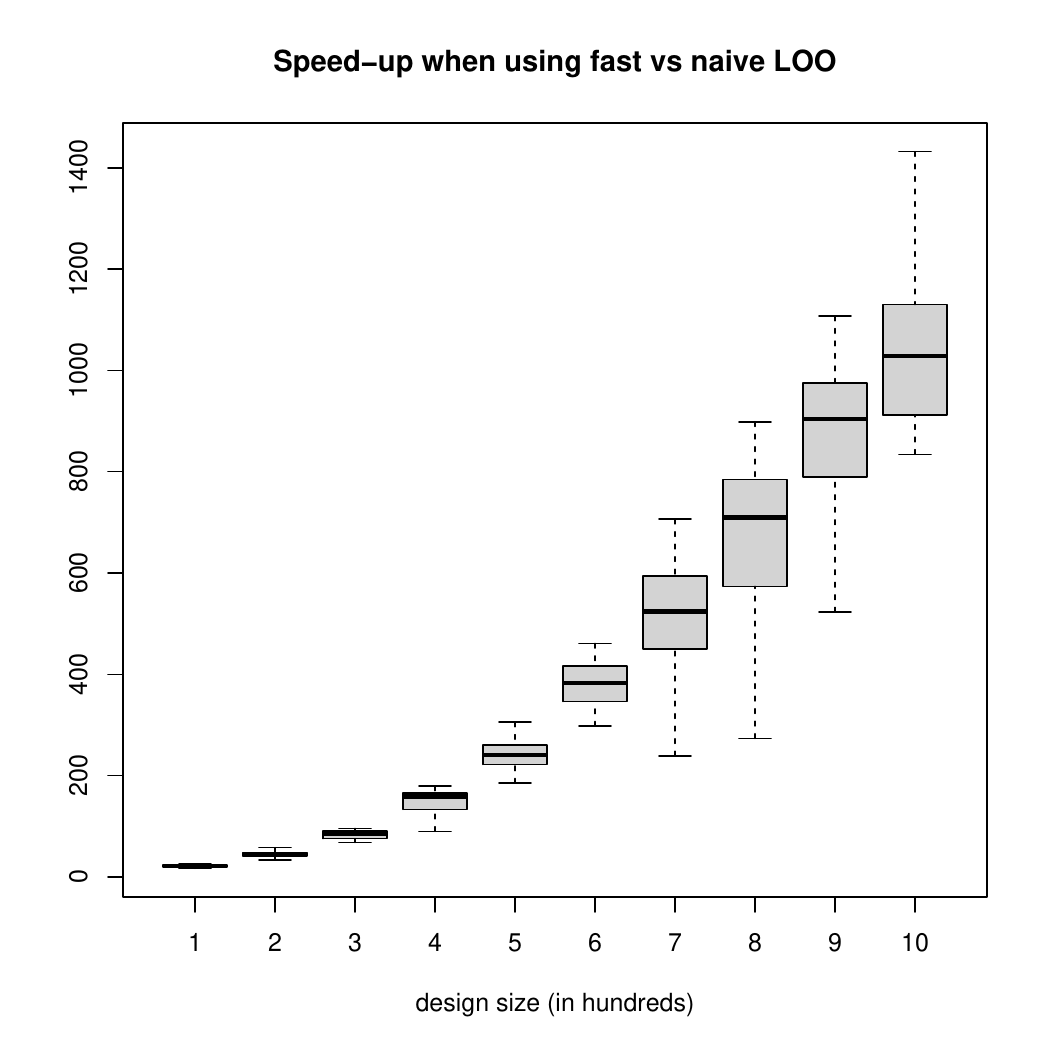}	
		\caption{Speed-up (ratio between times required to run the naive and fast methods) measured for LOO on $10$ regular designs, with $100$ to $1000$ points equidistributed on $[0,1]$, where each speed-up measure is repeated $50$ times.} 
		\label{speedupLOO}
	\end{center}
\end{figure}

We first consider Simple Kriging prediction based on $100$ points of our recurrent one-dimensional example test function, to which we apply LOO cross-validation with fast versus naive implementation. We find relative differences of the order of $10^{-14}$ and $10^{-12}$ when comparing vectors of leave-one-out predictions and variances obtained from the fast versus naive approaches, respectively (those relative differences consist of Eulidean norms of the differences divided by the norm of the relevant vector using the naive approach). In the current set-up, the fast implementation typically comes with a speed-up factor of $2$, a performance that is much improved when increasing the number of observations. Already with $1000$ observations, we found a speed-up of around $120$. 
Moving to a Universal Kriging model with quadratic trend (See Section~\ref{speedupLOO_UK} in Appendix for figures), we found with $100$ observations a speed-up nearing 5 while relative errors were still in the same tiny orders of magnitude. With $1000$ observations, we observed the same speed-up as in the case of Simple Kriging (with relative differences slightly increasing near $10^{-13}$ and $10^{-11}$).  

\bigskip

In Figure~\ref{speedupLOO} we present the speed-up measured for LOO on $10$ regular designs, with $100$ to $1000$ points equidistributed on $[0,1]$, where each experiment is repeated $50$ times by varying the seed. While the seed affects Gaussian Process model fitting, for each design and seed computation times of the fast versus naive LOO implementations are based on the same fitted model. Figure~\ref{speedupLOO_accu} represents the relative discrepancy between mean (resp. covariance) outputs of the two methods, with orders of magnitudes (in the $10^{-14}$ and $10^{-11}$, respectively) illustrating the accuracy of the fast method. 

\bigskip

Coming back to the speed-ups, note that costs that are compared here are not exactly those of calculating LOO outputs from scratch, as the Cholesky factor of the covariance matrix is already pre-calculated within the model fitting phase that is common to both procedures and the covariance matrix inversion at the heart of the fast method is henced facilitated (thanks to existing codes to obtain an inverse from a Cholesky factor).  While it would not make sense to count model fitting along the other operations as it involves the cumbersome task of estimating hyperparameters, we performed instead additional experiments where the covariance matrix was first rebuilt from the Cholesky factor and then factorized again, so as to be in an unfavourable situation for the fast method (as the rebuilt part comes as penalty resulting from the fact that the employed GP modelling code, the \textrm{km} function of the \textrm{R} package \textrm{DiceKriging}).
The corresponding results are plotted on Figure~\ref{speedupLOO_unfav} in Appendix, where it can be seen that in this unfavourable situation speed-ups are merely divided by a factor of $2$ but the immense benefit of using the fast LOO formula over the naive approach is not affected, all the more so that the design size increases.

\begin{figure}[h!]
	\begin{center}
		\includegraphics[width=.45\linewidth]{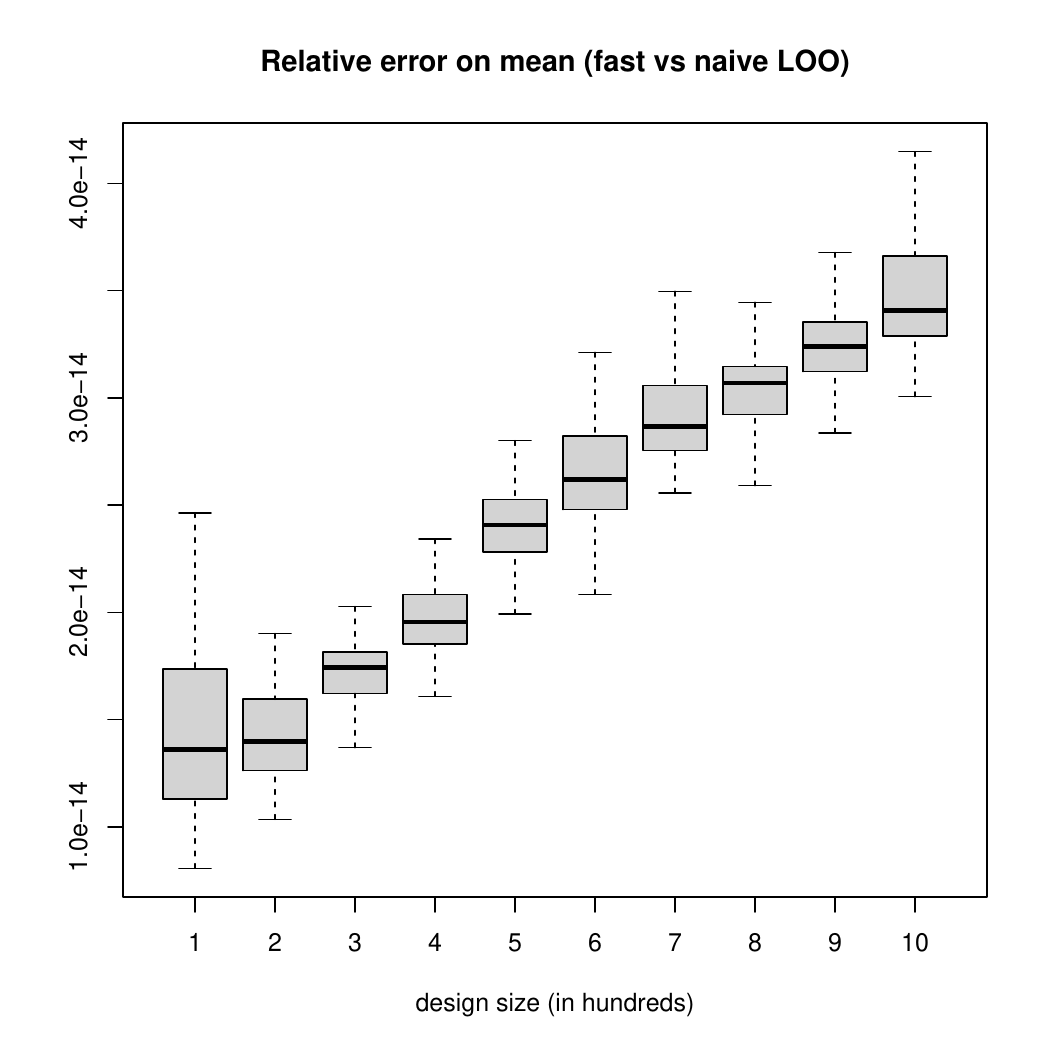}
		\includegraphics[width=.45\linewidth]{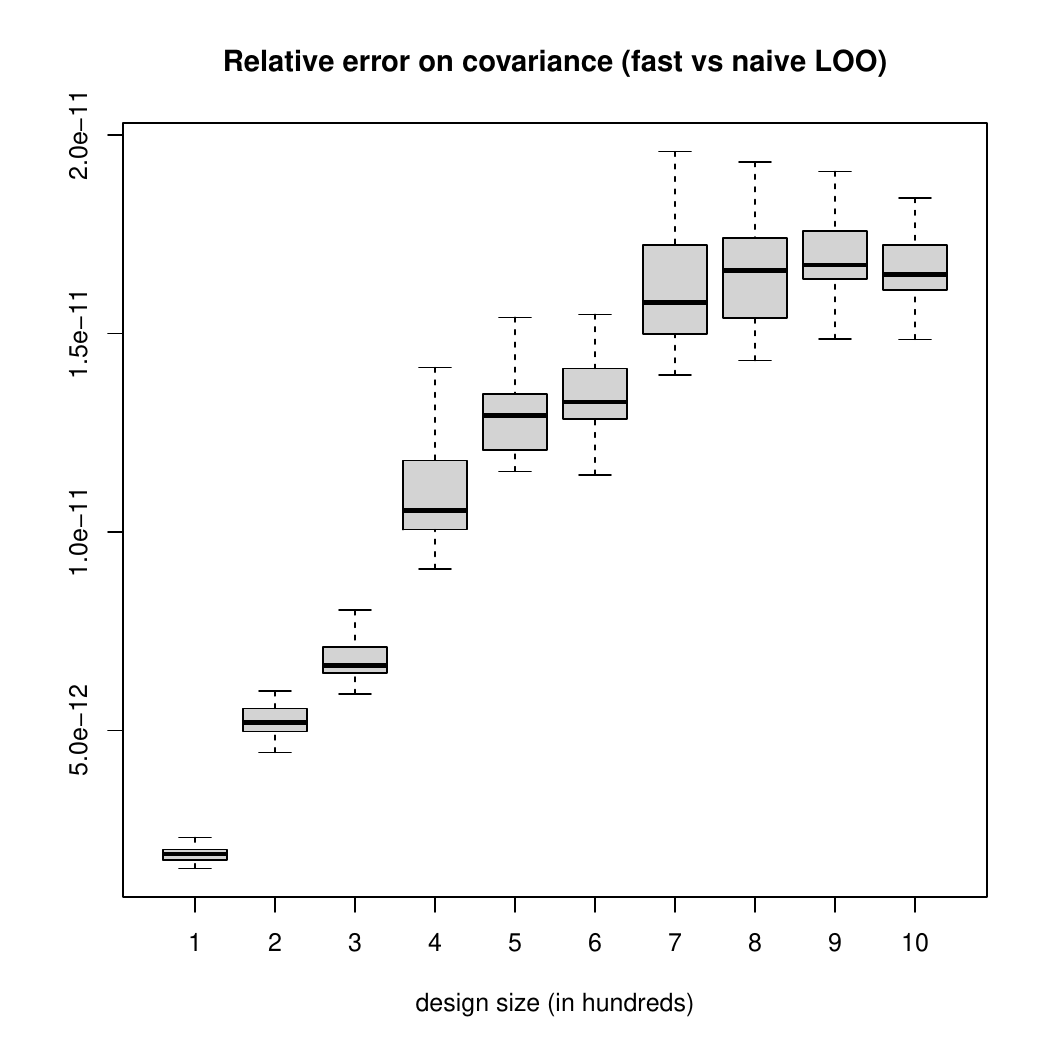}
		\caption{Relative errors on LOO means and covariances (between naive and fast methods, with the naive method as reference) measured as on Figure~\ref{speedupLOO}.}
		\label{speedupLOO_accu}
	\end{center}
\end{figure}

\begin{figure}[h!]
	\label{speedupLOO_unfav}
	\begin{center}
		\includegraphics[width=.65\linewidth]{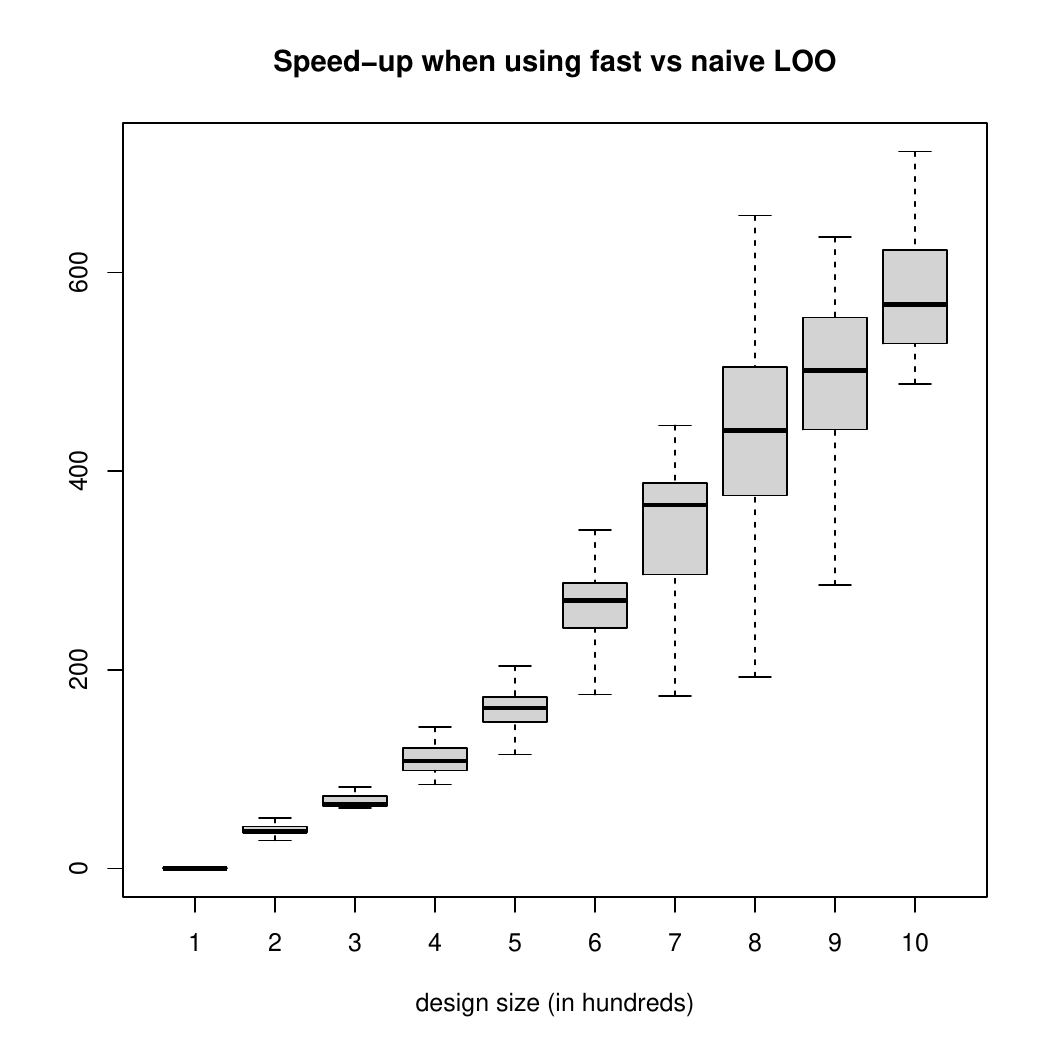}	
		\caption{Speed-up (ratio between times required to run the naive and fast methods) measured for LOO in an unfavourable situation where $K$ is re-built and its Cholesky factor is recalculated. Apart from that, the settings are the same as in  Figure~\ref{speedupLOO}.}
	\end{center}
\end{figure}

\newpage

\section{Supplementary numerical experiments}

\subsection{Multiple-fold cross-validation on the third example}
\label{CV_ex3}

\begin{figure}[h!]
	\label{CV_ex3}
	\begin{center}
		\includegraphics[width=.65\linewidth]{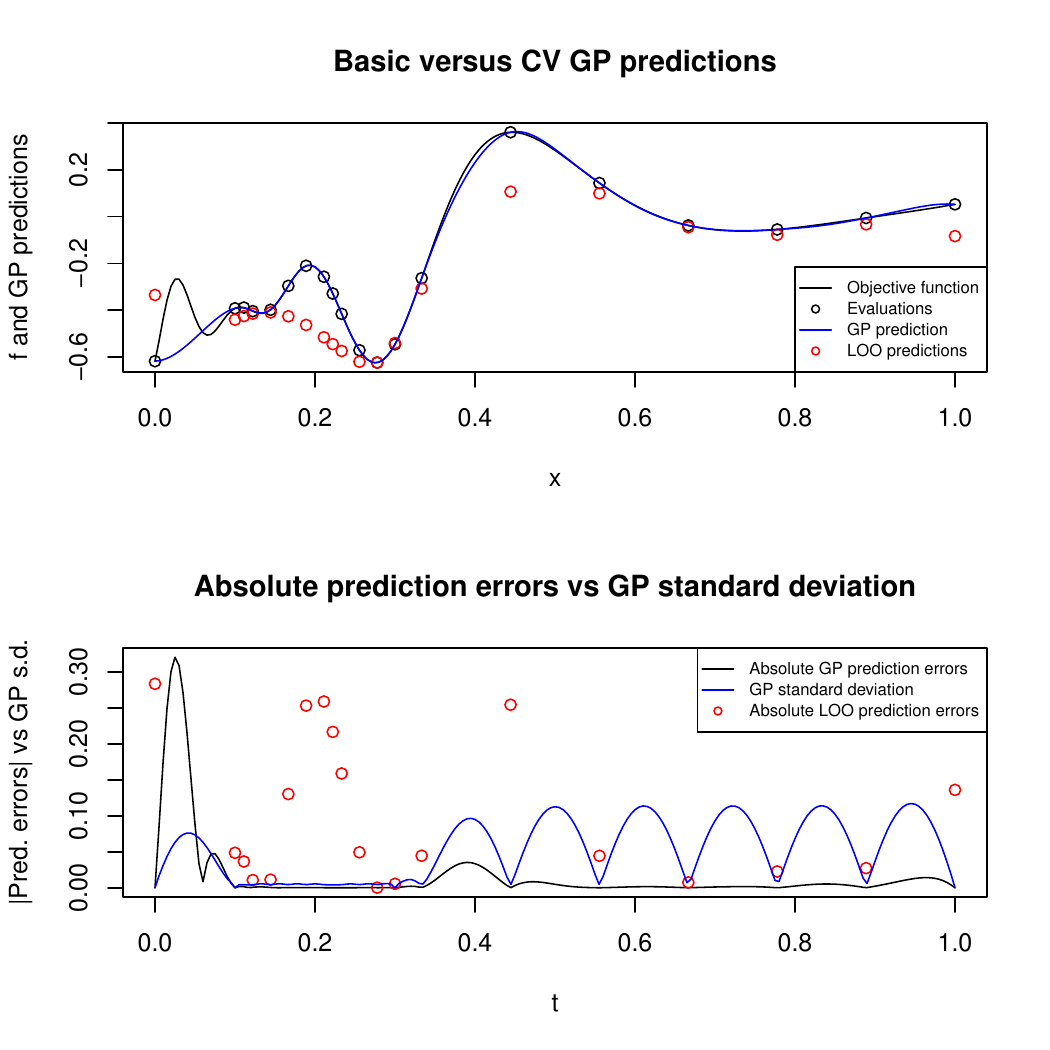}
		\caption{Upper panel: GP mean prediction (blue line) of the test function (black line) of Equation~\ref{xiong} based on evaluations at a regular grid complemented by clustered points between $0.1$ and $0.3$, CV predictions (red points). Lower panel: absolute prediction errors associated with GP (black line) and CV (red point) predictions, and GP prediction standard deviation (in blue). Here the second to eleventh points are grouped while the others are left as singletons (so $q=10$).}
	\end{center}
\end{figure}

\begin{figure}[h!]
	\label{xiong_QQplotsCV_ex3}
	\begin{center}
	\includegraphics[width=.4\linewidth]{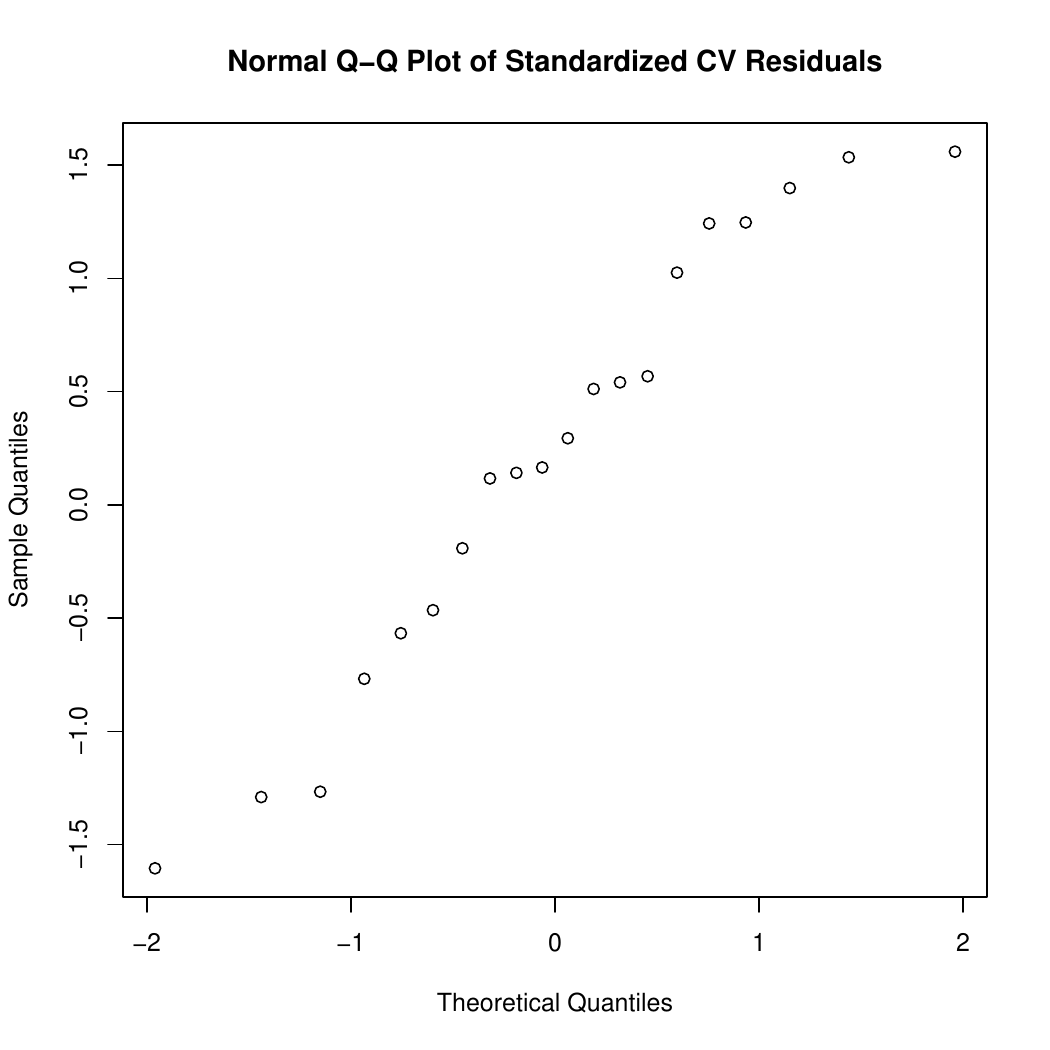}
	\includegraphics[width=.4\linewidth]{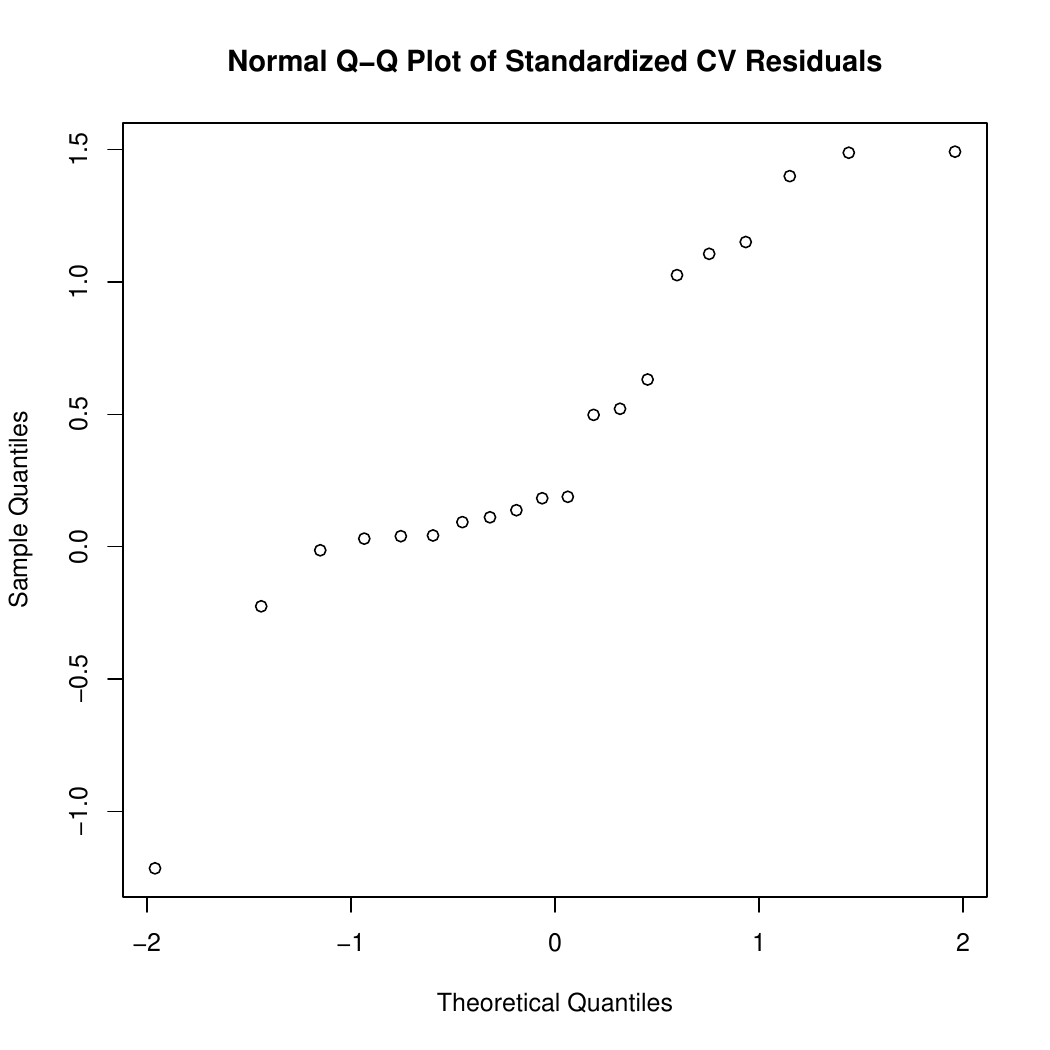}	
	\caption{On the effect of accounting for covariance in QQ-plots based on CV residuals for the example of Figure~\ref{CV_ex3}. Right panel: QQ-plot against $\mathcal{N}(0,1)$ of standardized CV residuals. Left panel: QQ-plot against $\mathcal{N}(0,1)$ of duly transformed CV residuals.}
	\end{center}
\end{figure}

\newpage 

\subsection{Fast versus naive LOO for Universal Kriging}
\label{speedupLOO_UK}

\begin{figure}[h!]
	\label{speedupLOO_UK}
	\begin{center}
		\includegraphics[width=.65\linewidth]{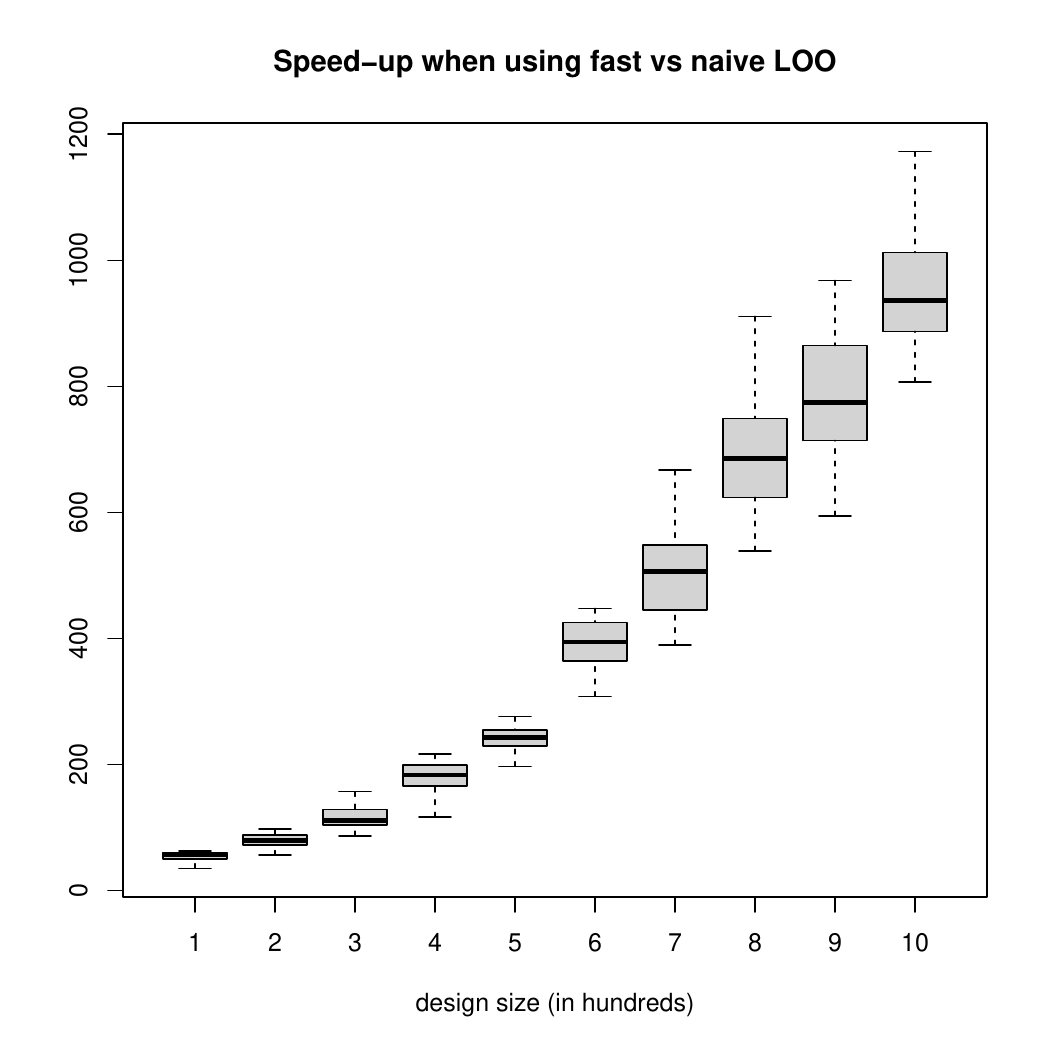}	
	\caption{Speed-up (ratio between times required to run the naive and fast methods) measured for LOO under Universal Kriging (with second order trend) on $10$ regular designs, with $100$ to $1000$ points equidistributed on $[0,1]$, where each measure is repeated $50$ times.} 
	\end{center}
\end{figure}

\begin{figure}[h!]
	\label{speedupLOO_UK__accu}
	\begin{center}
		\includegraphics[width=.45\linewidth]{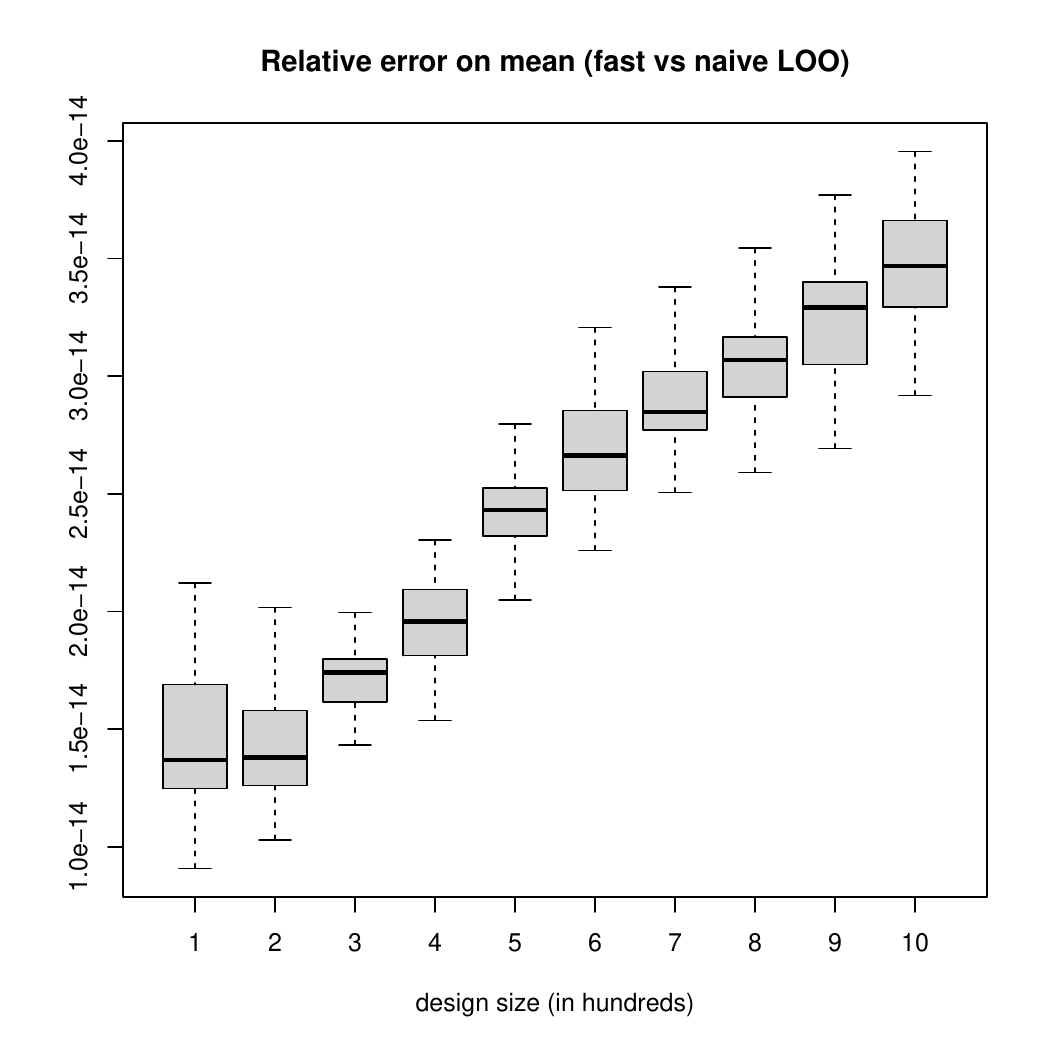}
		\includegraphics[width=.45\linewidth]{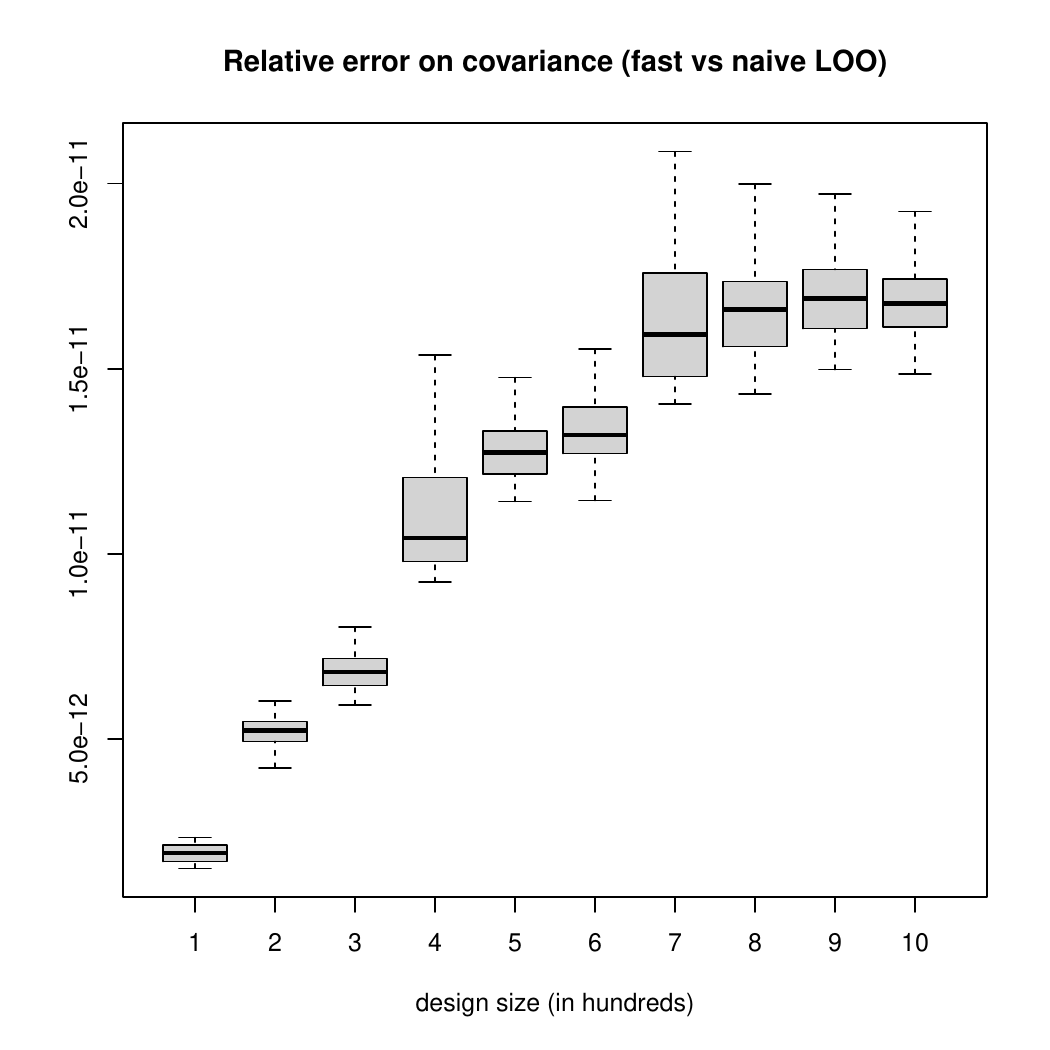}
		\caption{Relative errors on LOO means and covariances (between naive and fast methods, with the naive method as reference) measured as on Figure~\ref{speedupLOO_UK}.}
	\end{center}
\end{figure}

\end{document}